\documentclass[journal]{IEEEtran}
\usepackage{amsmath,amsfonts}
\usepackage{algorithm}
\usepackage{algpseudocode}
\usepackage{array}
\usepackage[caption=false,font=normalsize]{subfig}
\usepackage{textcomp}
\usepackage{stfloats}
\usepackage{url}
\usepackage{verbatim}
\usepackage{graphicx}
\usepackage{cite}
\usepackage{amsmath}
\usepackage{amssymb}
\usepackage{amsfonts}
\usepackage{amsthm}                    
\newtheorem{theorem}{Theorem}          
\usepackage{booktabs}
\begin{document}

\title{2D Coprime Pilots for Delay-Doppler Sensing \\ in OFDM-ISAC Systems} 

\author{
	\IEEEauthorblockN{ 
		Guanping Shang,
        Zeyan Zhuang,~\IEEEmembership{Graduate Student Member,~IEEE},
        Anzheng Tang,~\IEEEmembership{Graduate Student Member,~IEEE}, 
        Shenghui Song,~\IEEEmembership{Senior Member,~IEEE}, 
		and Chi-Ying Tsui,~\IEEEmembership{Senior Member,~IEEE}
	}
    \thanks{This work was supported in part by the Research Grants Council of Hong Kong under the General Research Fund (Project No. 16214125).}
	\thanks{G. Shang, Z. Zhuang, A. Tang, S. Song and C.Y. Tsui are with the Department of Electronic and 
    Computer Engineering, The Hong Kong University of Science and Technology, 
    Hong Kong, China. S. Song and C.Y. Tsui are also with the Division of Integrative Systems and Design, The Hong Kong University of Science and Technology, Hong Kong, China
    (e-mail: gshang@connect.ust.hk; zzhuangac@connect.ust.hk; eeaztang@ust.hk; eeshsong@ust.hk; eetsui@ust.hk).}

}

\maketitle

\begin{abstract}
Integrated Sensing and Communication (ISAC) is envisioned to endow future 6G systems with seamless sensing capabilities. To support efficient sensing with minimum communication overhead, sparse pilots embedded within communication frames have emerged as a promising solution. 
Along this line of research, existing studies have achieved engaging results in maximizing the unambiguous sensing region. However, jointly maximizing the sensing region and sensing accuracy remains challenging due to the lack of a unified performance metric and an effective pilot design framework. 
This paper jointly optimizes the unambiguous sensing region and sensing accuracy for estimating delay-Doppler (DD) parameters in Orthogonal Frequency Division Multiplexing (OFDM)-ISAC systems, where sensing mutual information (SMI) is adopted as a unified performance metric to characterize the overall sensing capability. Specifically, the joint optimization is formulated as an SMI maximization problem by systematically resolving sensing ambiguity and optimizing sensing accuracy. In particular, based on the generalized Bézout identity, we derive a 2D (time-frequency) coprime condition, which, as far as the authors know, is the first necessary and sufficient condition to achieve the unique estimation of DD parameters in the literature. Under this unambiguous condition, we further propose an Adam-Guided Iterative Refinement (AGIR) algorithm to optimize the sensing accuracy. Numerical results demonstrate the advantage of the proposed framework over existing designs, owing to the freedom offered by the 2D coprime condition in optimizing the sensing accuracy.
\end{abstract}


\section{Introduction}

\subsection{Background}

Integrated Sensing and Communication (ISAC) is regarded as a transformative paradigm for next-generation wireless networks by offering simultaneous data transmission and target sensing services \cite{intro1,intro2,5gresourceblocks}. To this end, a promising solution is to integrate sensing functionality into existing wireless networks \cite{ISACtoOFDM}. Given its superior auto-correlation properties, deterministic pilots embedded in communication, e.g., Orthogonal Frequency Division Multiplexing (OFDM), frames offer much more more favorable sensing properties than stochastic communication symbols, and the associated pilot design has attracted much research attention\cite{prsisac,pilottopayload}. However, allocating communication resources to sensing will inevitably degrade communication performance. As a result, it is crucial to minimize the number of sensing pilots placed in communication frames, resulting in sparse pilot patterns \cite{Unified}.  

Unfortunately, sparse sensing pilots will lead to sub-Nyquist sampling for delay and Doppler phase shifts, causing a severe reduction in the unambiguous region for target sensing \cite{fundamentaltradeoff,standardization6G,region}. Furthermore, sparse pilots inherently lead to limited observations for sensing targets, thereby significantly degrading the accuracy of parameter estimation \cite{lindstrom2026integrated,CS,2dotfs,liang2024}. Therefore, efficient sparse pilot designs that can jointly optimize the unambiguous sensing region and sensing accuracy are urgently needed.

\subsection{Related Works and Motivations}

To enlarge the unambiguous sensing region with sparse sampling, existing studies have developed structured pilot patterns based on coprime sampling principles \cite{initialcoprime,Silva2018,liuLowOverhead,Meiperiodiccoprime}. Specifically, the Chinese Remainder Theorem (CRT) has been widely utilized to design comb pilots with coprime sampling periods across different OFDM symbols for sensing purposes\cite{initialcoprime,Silva2018,liuLowOverhead}. Such designs enable unambiguous recovery of target parameters over an extended sensing region. Considering the periodic resource block structure, Mei \textit{et al.} \cite{Meiperiodiccoprime} further proposed an orthogonal pilot design with coprime intervals in both the time and frequency domains to achieve the maximum unambiguous delay-Doppler (DD) sensing region.
Nevertheless, these advances mainly addressed the sensing ambiguity issue, while the degradation of sensing accuracy induced by sparse pilot sampling remains insufficiently resolved\cite{lindstrom2026integrated}.

Realizing the joint optimization of the unambiguous sensing region and the sensing accuracy necessitates overcoming two fundamental challenges. 
Firstly, a unified performance metric that jointly characterizes the sensing region and sensing accuracy is still lacking. Although Wei \textit{et al.} \cite{separateanalysis} investigated both the unambiguous sensing region and the Cramér-Rao bound (CRB) for DD parameter estimation and insightfully revealed an inherent trade-off between ambiguity-free coverage and estimation accuracy, these two aspects were still evaluated using separate metrics. Thus, existing analyses cannot provide explicit guidance for balancing the region-accuracy trade-off, thereby preventing the pilot from approaching its theoretical limit of overall sensing capability \cite{Information-TheoreticLimits}.

Secondly, a pilot design framework that jointly optimizes sensing region and accuracy is needed. Existing methods typically achieve unambiguous sensing by rigidly placing pilots to satisfy coprime and guard-interval constraints in both time and frequency domains \cite{liuLowOverhead,Meiperiodiccoprime,separateanalysis}. However, the rigid constraints limit the freedom for sensing accuracy enhancement\cite{2dotfs,liang2024}. As a result, a more flexible pilot design framework is indispensable to overcome this structural rigidity and enable simultaneous improvement in sensing accuracy.

\IEEEpubidadjcol

In summary, jointly optimizing sensing region and accuracy necessitates a flexible pilot design framework with a unified performance metric. However, existing sensing performance metrics are mostly inherited from traditional radar systems \cite{mmWaveJRC}, measuring the autocorrelation and sidelobe ratios of signals\cite{ ISACwaveformreview,dhanasekaran, Ren, Superresolution}. Such fragmented metrics cannot provide a unified measurement of the sensing region and sensing accuracy.
To this end, sensing mutual information (SMI), which offers a rigorous information-theoretical measure for sensing performance, has recently been proposed for signal design in ISAC systems\cite{MIMOSMI,xie2024}. Existing works have leveraged SMI as the performance metric for sensing pilot design, primarily focusing on improving target detection or channel estimation performance\cite{hua2024, Bazzi2025}. 
Due to its capability in rigorously quantifying sensing uncertainty including both sensing region and accuracy, SMI will be utilized as the unified metric for pilot design. 

With SMI as the unified metric, this paper proposes a sparse pilot design framework to jointly optimize the sensing region and accuracy.  
For that purpose, we first derive the necessary and sufficient condition to achieve the maximum unambiguous sensing region.  
Under this condition, an Adam-Guided Iterative Refinement (AGIR) algorithm is further developed to maximize the sensing accuracy.

\subsection{Contributions}

The contributions of this work are summarized as follows:

\begin{itemize}
\item We establish the SMI between pilot patterns and the target's DD parameters, which measures both the unambiguous sensing region and sensing accuracy. With this unified metric, the joint pilot design is achieved through resolving sensing ambiguity and maximizing sensing accuracy.

\item To resolve the sensing ambiguity, we derive a 2D coprime condition based on the generalized Bézout identity over the integer lattices, which is necessary and sufficient for the unique estimation of the DD parameters. This condition guarantees unambiguous sensing with a more flexible pilot configuration than existing designs, leaving more freedom for sensing accuracy optimization.

\item Subject to the 2D coprime condition, the sensing accuracy optimization problem is formulated. An AGIR algorithm is then proposed to tackle the problem, where the Bayesian Cramér-Rao bound (BCRB) is utilized as a tractable surrogate for optimizing sensing accuracy within the unambiguous sensing region. 

\end{itemize}
Numerical results demonstrate the superiority of the proposed framework over existing baseline designs and highlight the flexibility of the 2D coprime condition in pilot design over existing solutions.

\subsection{Paper Structure and Notations}
\textit{Notations:} Throughout this paper, scalars are denoted by italic letters, while vectors and matrices are represented by boldface lowercase and uppercase letters. The operators $(\cdot)^{T}$ and $(\cdot)^{H}$ denote the transpose and conjugate transpose, respectively. The Hadamard product is represented by $\odot$, and $\mathbb{E}[\cdot]$ denotes the expectation operator. The notations $\mathbb{C}^{M \times N}$, $\mathbb{R}^{M \times N}$, and $\mathbb{Z}^{M \times N}$ denote the set of $M \times N$ complex, real, and integer matrices, respectively.  $\mathcal{CN}(\boldsymbol{\mu}, \mathbf{\Sigma})$ represents the complex Gaussian distribution with mean $\boldsymbol{\mu}$ and covariance $\mathbf{\Sigma}$. $\mathbf{I}_{D}$ refers to the $D \times D$ identity matrix. $\det(\cdot)$ and $\text{gcd}(\cdot)$ denote the matrix determinant and the greatest common divisor, respectively.

\textit{Paper Structure:} The rest of the paper is organized as follows:  
Sec. II introduces the system model. 
Sec. III introduces SMI as the unified performance metric and transforms its maximization into resolving sensing ambiguity and maximizing sensing accuracy. 
Sec. IV derives the 2D coprime condition for resolving sensing ambiguity and presents the AGIR algorithm to further optimize sensing accuracy.
Sec. V reports the experimental results, and Sec. VI concludes the paper.

\section{System Model}

\label{sec:systemmodelforSMI}

As illustrated in Fig.~\ref{senario}, we consider a bistatic OFDM-ISAC system, where a dual-functional transmitter and a sensing receiver jointly support communication and sensing services. Specifically, sensing pilots are periodically embedded within the OFDM frames for sensing purposes. The target-reflected echoes are processed by the sensing receiver for parameter estimation.

\begin{figure}[]
\centering
\includegraphics[width=3.4in]{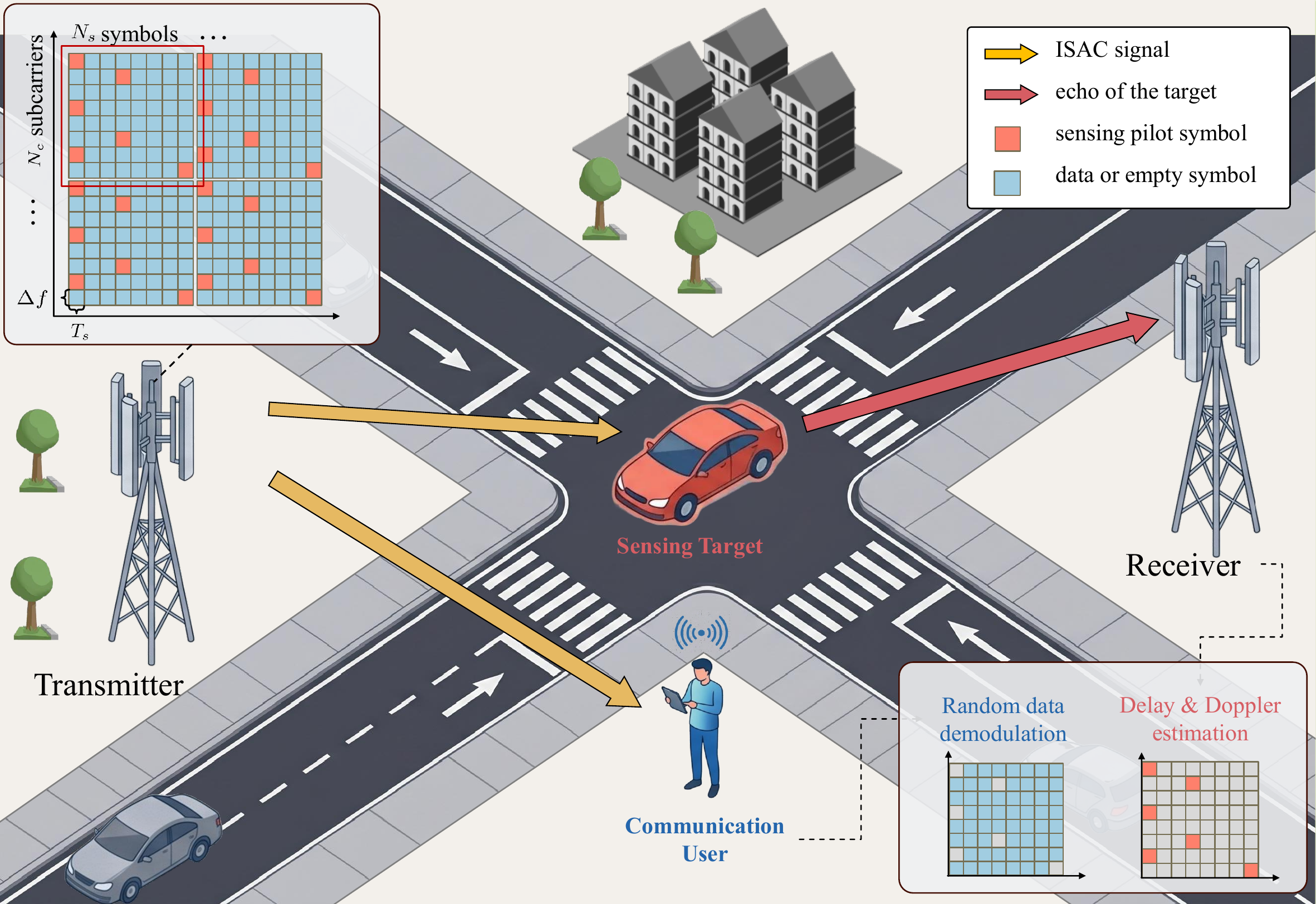}
\caption{Illustration of an OFDM-based ISAC system with 
periodic sparse pilots for joint communication and sensing.}
\label{senario}
\end{figure}

\subsection{Sensing Model}
Without loss of generality, we consider an OFDM frame composed of $B_t \times B_f$ time-frequency resource blocks, where each resource block has $N_c$ subcarriers and $N_s$ symbols \cite{Meiperiodiccoprime}. Let $ b_f$ and $ b_t$ denote the frequency and time indices of resource blocks with $0\leq b_f\leq B_f-1$ and $0\leq b_t\leq B_t-1$, respectively.

Define $\mathcal{P}=\{(m_p,n_p)\mid p \in {1,...,P} \}$ as the pilot set in one resource block, where $ m_p$ and $ n_p$ denote the subcarrier and symbol index with $ 0 \leq m_p \leq N_c-1$ and $0 \leq n_p \leq  N_s-1 $, respectively, and $P=|\mathcal{P}|$ represents the number of pilot symbols in each block. Assuming the pilot patterns are identical across all resource blocks, the index set for all pilot symbols can be expressed as:
\begin{equation}
\label{Pmn}
\begin{aligned}
\mathcal{R} = \bigcup_{(b_f,b_t)} \{ (b_f N_c + m_p,b_t N_s+n_p) |  \forall (m_p,n_p) \in \mathcal{P}  \}. 
\end{aligned}
\end{equation}  
Let $\mathbf{X}_\text{RS} \in \mathbb{C}^{N_c \times N_s}$ denote the pilot matrix in one resource block with its $(m,n)$-th entry given by:
\begin{equation}
\mathbf{X}_\text{RS}(m,n)=
\begin{cases}
1, & (m,n)\in \mathcal{P},\\
0, & (m,n)\notin \mathcal{P}.
\end{cases} 
\end{equation}
Thus, the transmitted baseband pilot signal for the whole frame can be given by:
\begin{equation}\label{eq:tx_signal_re1}
\begin{aligned}
s_p(t) = &
\sum_{b_t =0}^{B_t-1}
\sum_{b_f=0}^{B_f-1}
\sum_{(m_p, n_p)\in \mathcal{P}}
\mathbf{X}_\text{RS} (m_p,n_p) e^{j2\pi m_p\Delta f t}\\
&e^{j2\pi b_fN_c\Delta f t} \cdot  \mathrm{rect}\!\left(\frac{t-(b_tN_s+n_p)T_s}{T_s}\right), 
\end{aligned}
\end{equation}
where $T_s$ represents the symbol duration and $\Delta f=1/T_s$ denotes the subcarrier spacing. $\mathrm{rect}(\cdot)$ refers to the rectangular window function.

The sensing channel for a target is characterized by the DD parameters $\boldsymbol{\theta}=[\tau, \nu]^T$, where $\tau$ and $\nu$ denote the delay and Doppler shift of the target, respectively. 
With a sufficiently long cyclic prefix, the inter-symbol interference can be neglected. 
Under such circumstances, the channel response for the pilot symbols in the $(b_f, b_t)$-th resource block can be denoted as $\mathbf{\Phi}_{\tau, \nu} \in \mathbb{C}^{N_c \times N_s}$ with its $(m, n)$-th element given by: 
\begin{equation}
\label{phiij}
\mathbf{\Phi}_{\tau, \nu} (m, n) = e^{-j2\pi \Delta f(m+b_fN_c)\tau} \cdot e^{j2\pi T_s (n + b_tN_s)\nu}. 
\end{equation}
Accordingly, the received pilot signal  $\mathbf{Y}_\text{RS} \in \mathbb{C}^{N_c \times N_s}$ for one resource block can be expressed as: 
\begin{equation}
\label{eq:rx_signal_re1}
\mathbf{Y}_\text{RS} = \xi \mathbf{\Phi}_{\tau, \nu}  \odot \mathbf{X}_\text{RS} + \mathbf{W},
\end{equation}
where $\xi \in \mathbb{C}$ represents the complex propagation coefficient, and $\mathbf{W} \in \mathbb{C}^{N_c \times N_s}$ denotes the equivalent additive white Gaussian noise with independent and identically distributed (i.i.d.) entries following $\mathcal{CN}(0, \sigma^2)$. For ease of illustration, the residual inter-carrier interference is treated as a random disturbance and incorporated into the noise term $\mathbf{W}$.

\subsection{Sensing Region and Accuracy with Sparse Sampling}

The maximum unambiguous sensing region for DD parameters is determined by the subcarrier spacing ${\Delta f}$ and symbol duration $T_s$ as $\tau_{\max}= \frac{1}{\Delta f}$ and $\nu_{\max}= \frac{1}{T_s}$. This can be achieved with full sampling, i.e., all subcarriers and symbols are utilized for sensing. 
However, the sparse configuration of pilots $\mathbf{X}_\text{RS}$ limits the maximum observable delay and Doppler regions, resulting in ambiguity in DD domain\cite{Silva2018}. As illustrated in Fig.~\ref{fig:ambiguity}, the sub-Nyquist sampling induces a one-to-many mapping from the parameter space $\Theta \subset \mathbb R^{2}$ (upper sub-figure) to the estimation result based on $\mathbf{Y}_\text{RS}$ (lower sub-figure). As a result, the maximum unambiguous sensing region is partitioned into indistinguishable sub-regions $\tau_{\max}^\text{sub}$ and $\nu_{\max}^\text{sub}$, causing ambiguity in estimating DD parameters. 
Furthermore, sparse sampling inevitably leads to limited bandwidth and duration and thus degraded sensing accuracy. Such parameter estimation errors around the ground truth value are referred to as local uncertainty in Fig.~\ref{fig:ambiguity}.

As mentioned above, both unambiguous sensing region and sensing accuracy must be jointly optimized to achieve satisfactory sensing performance. For that purpose, we need a unified performance metric to jointly measure sensing region and accuracy. Unfortunately, traditional metrics were designed to evaluate either sensing region or sensing accuracy and fail to capture their coupled relationship for a holistic assessment of the overall sensing capability. 
In the next section, we propose to utilize SMI as a unified performance metric.

\begin{figure}[tbp]
\centering 
\includegraphics[width=3.5in]{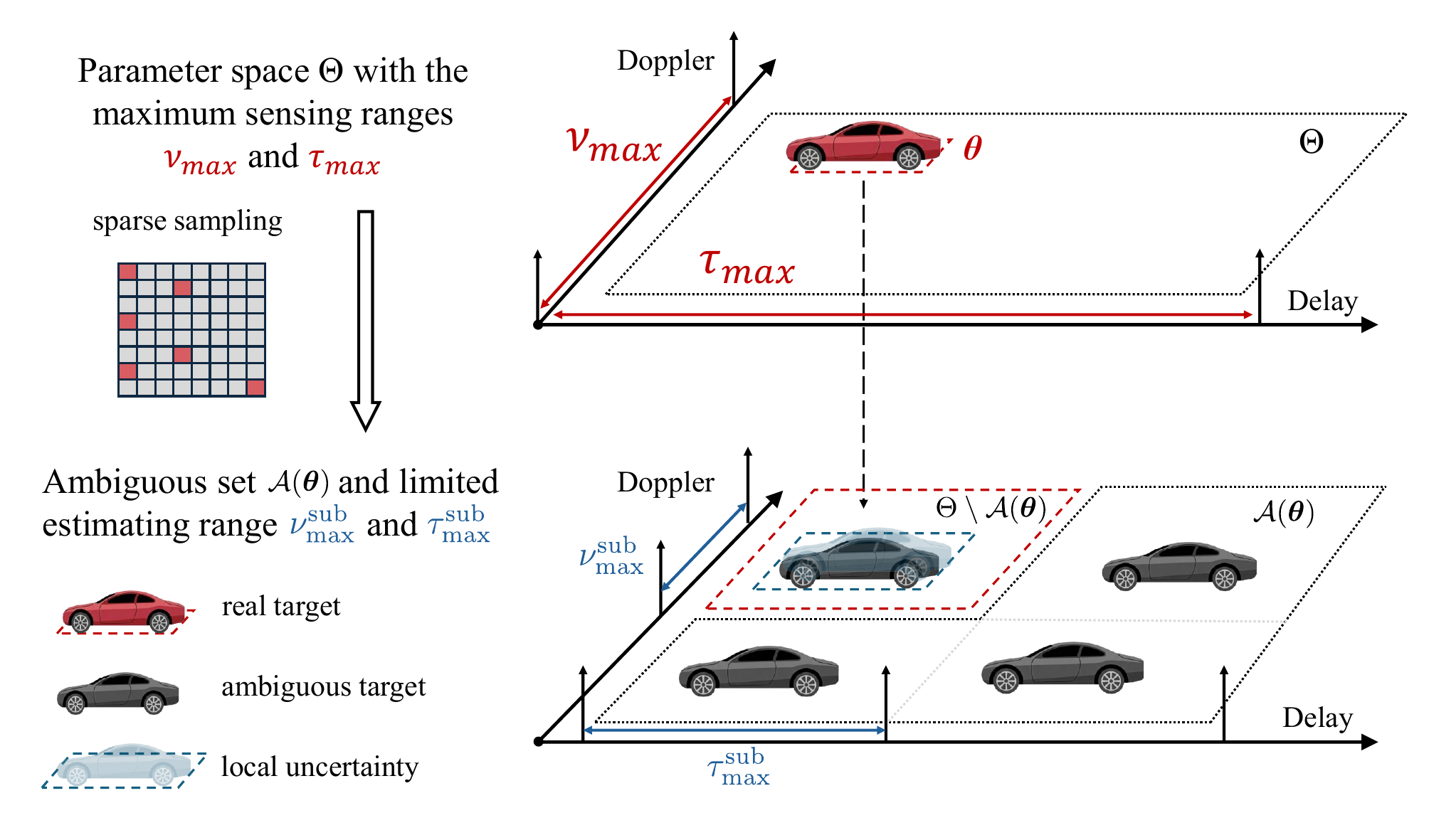}
\caption{
Sparse sampling causes the one-to-many mapping from the parameter space $\Theta$ to the estimation result, which results in ambiguous sensing region and local uncertainty.}
\label{fig:ambiguity}
\end{figure}

\section{SMI Analysis and Problem Formulation}
\label{sec:SMIanaly}

In this section, we establish the SMI between the received pilot $\mathbf{Y}_\text{RS}$ and the DD parameters $\boldsymbol{\theta}$, and explain why it is capable to capture the joint behavior of sensing region and accuracy. Then, we formulate the SMI maximization problem to resolve sensing ambiguity and optimize sensing accuracy.
\subsection{Sensing Mutual Information}

The SMI between  $\boldsymbol{\theta}$ and $\mathbf{Y}_\text{RS}$ can be given by: 
\begin{equation}
\begin{aligned}
I(\mathbf{Y}_\text{RS};\boldsymbol{\theta})
&= h(\mathbf{Y}_\text{RS})-h(\mathbf{Y}_\text{RS}|\boldsymbol{\theta})\\
&=\mathbb{E}_{\boldsymbol{ \theta},\mathbf{Y}_\text{RS}}
\Big[\log{\frac{p(\boldsymbol{\mathbf{Y}_\text{RS}|\theta})}{\int_{\Theta } p(\boldsymbol{\mathbf{Y}_\text{RS}|\vartheta }) 
p(\boldsymbol{\vartheta }) \,d\boldsymbol{\vartheta } }} \Big], 
\end{aligned}
\label{smi1}
\end{equation}
where $h(\cdot)$ denotes the differential entropy for continuous random variables, and $p(\cdot)$ represents the probability density function (PDF). 
By statistically measuring the uncertainty reduction in $\boldsymbol{\theta}$ after knowing $\mathbf{Y}_\text{RS}$, $I(\mathbf{Y}_\text{RS};\boldsymbol{\theta})$ inherently measures both the unambiguous sensing region and sensing accuracy, thereby capable of serving as a unified optimization metric. More detailed justification will be given in Sec.~\ref{Unif-subsec}. 

Denoting $N_\text{RS}$ as the maximum available number of sensing pilots in one resource block, the pilot design can thus be formulated as an optimization problem to maximize the SMI $I(\mathbf{Y}_\text{RS};\boldsymbol{\theta})$ for a representative resource block:  
\begin{equation}
\label{problemformulation0}
\begin{aligned}
\underset{\mathcal{P}}{\text{max}}\quad &I(\mathbf{Y}_\text{RS};\boldsymbol{\theta})\\
\text{s.t.} \quad& P= N_\text{RS}. 
\end{aligned}
\end{equation}

\subsection{The Unified Metric for Sensing Region and Accuracy}\label{Unif-subsec}

Sparse sampling will cause ambiguity between the ground truth $\boldsymbol{\theta}$ and any $\boldsymbol{\theta}'$ that satisfies $p(\mathbf{Y}_\text{RS} \mid \boldsymbol{\theta}) = p(\mathbf{Y}_\text{RS}\mid \boldsymbol{\theta}')$ for any realization of $\mathbf{Y}_\text{RS}$. Specifically, the sensing receiver will not be able to differentiate between $\boldsymbol{\theta}$ and $\boldsymbol{\theta}'$.
As illustrated in Fig.~\ref{fig:ambiguity}, we can define the ambiguity set for a given $\boldsymbol{\theta}$ as: 
\begin{equation}\label{eq:ambiguity_set}
\mathcal{A}(\boldsymbol{\theta}) \triangleq \left\{ \boldsymbol{\theta}' \in \Theta : 
p(\mathbf{Y}_\text{RS} \mid \boldsymbol{\theta}') = p(\mathbf{Y}_\text{RS}\mid \boldsymbol{\theta}),  \boldsymbol{\theta}'\neq \boldsymbol{\theta}  \right\}.
\end{equation}
Thus, the denominator in the second line of \eqref{smi1} can be decomposed into two integrals with respect to $\mathcal{A}(\boldsymbol{\theta})$ and its complement:  
\begin{equation}
\begin{aligned}
\label{smi2}
&I(\mathbf{Y}_\text{RS};\boldsymbol{\theta}) =\mathbb{E}_{\boldsymbol{ \theta},\mathbf{Y}_\text{RS}} \log \\
&\Big[
    \frac{  p(\mathbf{Y}_\text{RS} \mid \boldsymbol{\theta})}{
    \int_{\mathcal{A}(\boldsymbol{\theta})} 
    p(\mathbf{Y}_\text{RS} \mid \boldsymbol{\vartheta}) 
    p(\boldsymbol{\vartheta})\, d\boldsymbol{\vartheta}+
    \int_{\Theta \setminus \mathcal{A}(\boldsymbol{\theta})} 
    p(\mathbf{Y}_\text{RS} \mid \boldsymbol{\vartheta}) 
    p(\boldsymbol{\vartheta})\, d\boldsymbol{\vartheta}
  } \Big]. 
\end{aligned}
\end{equation}
The first integral is performed within the ambiguous set and can be simplified as: 
\begin{equation}
\int_{\mathcal{A}(\boldsymbol{\theta})} p(\boldsymbol{\mathbf{Y}_\text{RS}|\vartheta}) 
p(\boldsymbol{\vartheta}) \, d\boldsymbol{\vartheta} \overset{\text{(a)}}{=}
p(\boldsymbol{\mathbf{Y}_\text{RS}|\theta}) 
\int_{\mathcal{A}(\boldsymbol{\theta})} p(\boldsymbol{\vartheta}) \, d\boldsymbol{\vartheta},  
\end{equation}
where (a) follows from $p(\boldsymbol{\mathbf{Y}_\text{RS}}|\boldsymbol{\vartheta})=p(\boldsymbol{\mathbf{Y}_\text{RS}}|\boldsymbol{\theta})$ for $\boldsymbol{\vartheta} \in \mathcal{A}(\boldsymbol{\theta})$.  
By dividing $p(\mathbf{Y}_\text{RS} \mid \boldsymbol{\theta})$ on both the numerator and denominator of \eqref{smi2}, we can rewrite the SMI as: 
\begin{equation}
\begin{aligned}
\label{smi3}
&I(\mathbf{Y}_\text{RS};\boldsymbol{\theta}) =\\
&\mathbb{E}_{\boldsymbol{ \theta},\mathbf{Y}_\text{RS}}\log \Bigg[
    \frac{1}{
  \underbrace{\pi (\boldsymbol{\theta})}_\text{ambiguity}+
\underbrace{  \frac{
    \int_{\Theta \setminus \mathcal{A}(\boldsymbol{\theta})} 
    p(\mathbf{Y}_\text{RS} \mid \boldsymbol{\vartheta}) 
    p(\boldsymbol{\vartheta})\, d\boldsymbol{\vartheta}
  }{
    p(\mathbf{Y}_\text{RS} \mid \boldsymbol{\theta})
  }}_{\text{uncertainty}} } \Bigg], 
\end{aligned}
\end{equation}
where $\pi (\boldsymbol{\theta})=\int_{\mathcal{A}(\boldsymbol{\theta})} p(\boldsymbol{\vartheta})\, d\boldsymbol{\vartheta}$ represents the probability of $\boldsymbol{\theta}$ falling into the ambiguity set $\mathcal{A}(\boldsymbol{\theta})$. 

\textbf{Remark 1:}
\textit{
It can be observed from \eqref{smi3} that SMI can be decomposed into two adverse factors: the sensing ambiguity due to the ambiguity set $\mathcal{A}(\boldsymbol{\theta})$ and the uncertainty within the unambiguous sensing region $\Theta \setminus \mathcal{A}(\boldsymbol{\theta})$. 
This indicates that SMI is capable of measuring both the unambiguous sensing region and sensing accuracy.}

\subsection{Problem Formulation} 
\label{secproblemformulation}

\begin{figure}[tbp]
\centering
\includegraphics[width=3.4in]{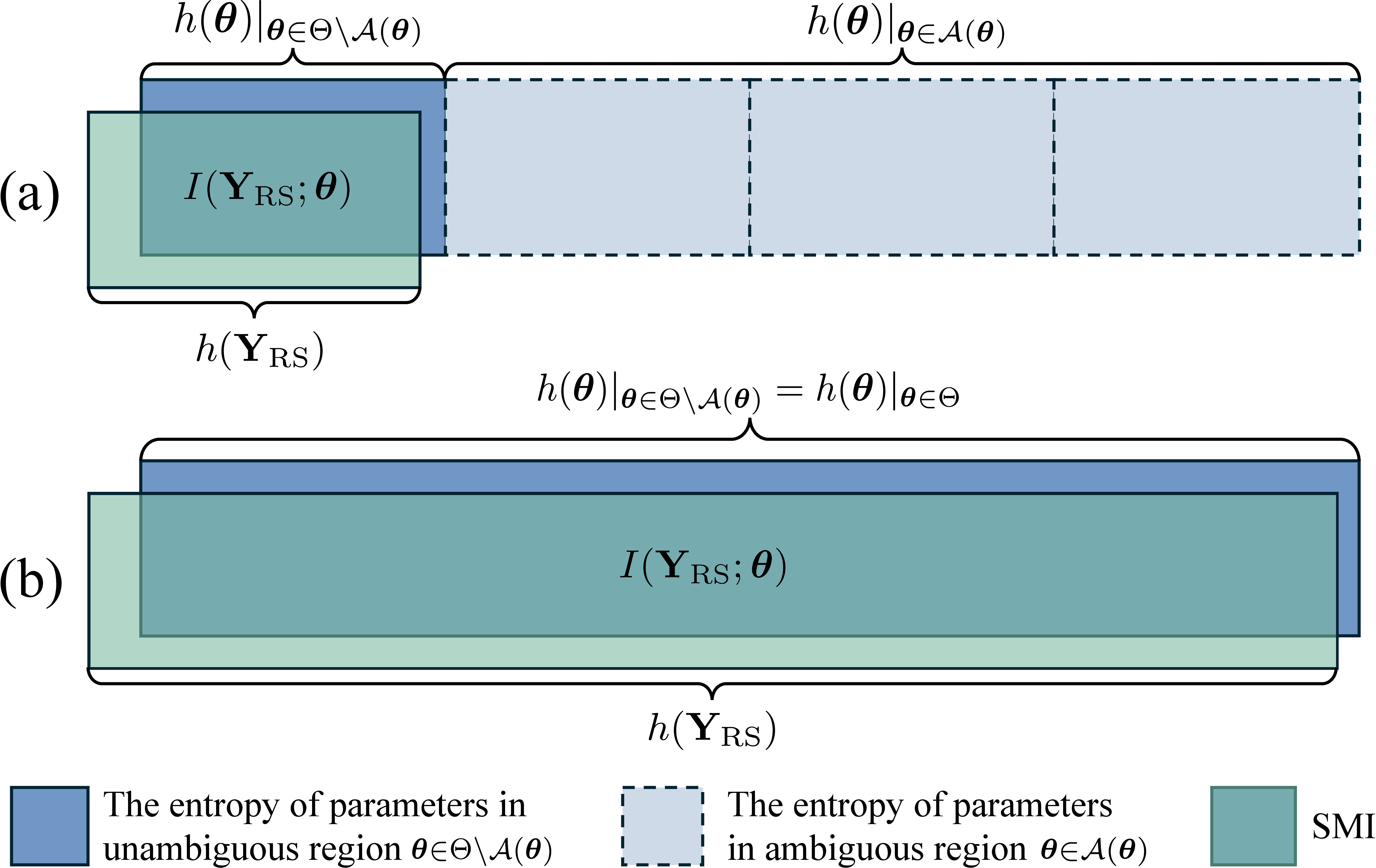}
\caption{Relationship between SMI and $h(\boldsymbol{\theta})$: (a) The presence of sensing ambiguity imposes an irreducible entropy loss for achievable SMI. (b) Resolving sensing ambiguity removes the loss, enabling the SMI to approach the maximum value.
Here, $h(X)|_\mathcal{B} = - \int_\mathcal{B} p(x) \log p(x) dx $ denotes the entropy of a random variable $X$ within a region $\mathcal{B}$.}
\label{ven}
\end{figure}

As mentioned above, the achievable SMI relies on both sensing accuracy and sensing region. We now look into the respective contributions of the two factors. With a given unambiguous sensing region, sensing uncertainty can be resolved as the Signal-to-Noise Ratio (SNR) increases. 
Specifically, as SNR increases, the SMI will approach the entropy of $\boldsymbol{\theta}$ within the unambiguous sensing region, namely the restricted entropy $h({\boldsymbol{\theta}}) |_{\boldsymbol{\theta} \in \Theta \setminus \mathcal{A}(\boldsymbol{\theta}) }$, as shown by the dark-blue block in Fig.~\ref{ven}a. 
However, there remains an irreducible entropy loss $h(\boldsymbol{\theta})|_{\boldsymbol{\theta} \in  \mathcal{A}(\boldsymbol{\theta})}$ (the light-blue blocks in Fig.~\ref{ven}a) due to the existence of ambiguous regions. 
This indicates that sensing ambiguity fundamentally limits the achievable SMI, and the maximum SMI can only be attained when the ambiguity is resolved, as illustrated in Fig.~\ref{ven}b. In other words, sensing ambiguity determines the upper bound of SMI, while improving sensing accuracy enables approaching this bound. Therefore, we first resolve the sensing ambiguity and then optimize the sensing accuracy.

According to \eqref{smi3}, eliminating the ambiguity term $\pi (\boldsymbol{\theta})$ is equivalent to enforcing an empty ambiguity set ($\mathcal{A}(\boldsymbol{\theta}) = \emptyset$). This condition restricts the pilot pattern $\mathcal{P}$ to the following feasible set:  
\begin{equation}
\Omega \triangleq \left\{ \mathcal{P} : \mathcal{A}(\boldsymbol{\theta}) = \boldsymbol{\emptyset}, \ \forall \boldsymbol{\theta} \in \Theta\right\}.    
\end{equation}
Then, we can maximize the sensing accuracy under the unambiguous constraint:  
\begin{equation}
\label{problemformulation}
\begin{aligned}
\underset{\mathcal{P}}{\text{max}}
\ \mathcal{U(P)}=
\mathbb{E}_{\boldsymbol{ \theta},\mathbf{Y}_\text{RS}} &\log\left[
    \frac{p(\mathbf{Y}_\text{RS} \mid \boldsymbol{\theta}; \mathcal{P})}
    {\int_{\Theta \setminus \mathcal{A}(\boldsymbol{\theta})} 
        p(\mathbf{Y}_\text{RS} \mid \boldsymbol{\vartheta}; \mathcal{P}) 
        p(\boldsymbol{\vartheta}) \, d\boldsymbol{\vartheta}}
\right],\\
\text{s.t.}  \quad& \mathcal{P} \in \Omega, \\
& P = N_\text{RS}, 
\end{aligned}
\end{equation}
where the objective $\mathcal{U}(\mathcal{P})$ measures the sensing accuracy as the inverse of the local uncertainty.

\section{Coprime Pilots for Delay-Doppler Sensing}
\label{sec:1d_model}

This section introduces the pilot design framework to maximize the sensing accuracy while resolving the sensing ambiguity. Specifically, in Sec.~\ref{sec:1d_model}-A, we derive a 2D coprime condition for unambiguous sensing, based on the generalized Bézout identity. Subsequently, we provide an insight into the sensing region and sensing accuracy by investigating the Ambiguity Function (AF) in Sec.~\ref{sec:1d_model}-B. Motivated by the insight from the above analysis, in Sec.~\ref{sec:1d_model}-C we use BCRB as a surrogate for measuring sensing accuracy, and in Sec.~\ref{sec:1d_model}-D we propose an AGIR algorithm to maximize sensing accuracy under 2D coprime constraints.

\subsection{Coprime Condition for Resolving Sensing Ambiguity}
\label{subsec:1d_delay}

We first characterize sensing ambiguity by a system of integer equations with respect to DD parameters. Subsequently, based on the generalized Bézout's identity, we establish the necessary and sufficient condition to guarantee a unique solution to the system.

\subsubsection{Ambiguity Regions}
With the periodic structures across different resource blocks, the maximal time and frequency sampling periods are equal to the inter-block intervals $N_c$ and $N_s$, respectively. Consequently, the minimal unambiguous regions for delay and Doppler can be given by $\tau_{\min}^\text{sub}= \frac{1}{N_c \Delta f}$ and $\nu_{\min}^\text{sub}=  \frac{1}{N_s T_s}$, respectively \cite{Silva2018}. 

Fig.~\ref{integers} illustrates the minimal unambiguous regions caused by the periodic sampling over resource blocks. It can be observed from Fig.~\ref{integers} that, although the receiver can uniquely resolve the fractional components $\tau^r$ (and $\nu^r$) of the DD parameters within the minimal unambiguous region defined by $\tau_{\min}^\text{sub}$ (and $\nu_{\min}^\text{sub}$), it suffers from severe ambiguity regarding which sub-region the target belongs to. Specifically, the target can be interpreted at multiple locations in DD domain:  
\begin{equation}
\label{adelayDoppleramb}
\begin{cases}
\tau= \tau^r + \frac{q_\tau}{N_c \Delta f}, \\
\nu= \nu^r + \frac{q_\nu}{N_s T_s},  
\end{cases}
\end{equation}
where $q_\tau \in \{0, 1, \dots, N_c-1\}$ and $q_\nu \in \{0, 1, \dots, N_s-1\}$ are referred to as the ambiguity integer for delay and Doppler, respectively. In other words, the receiver cannot distinguish the multiple equivalent delay and Doppler pairs in \eqref{adelayDoppleramb} due to the confusion caused by ambiguity integers $q_\tau$ and $q_\nu$.

\begin{figure}[t]
\centering
\includegraphics[width=3.45in]{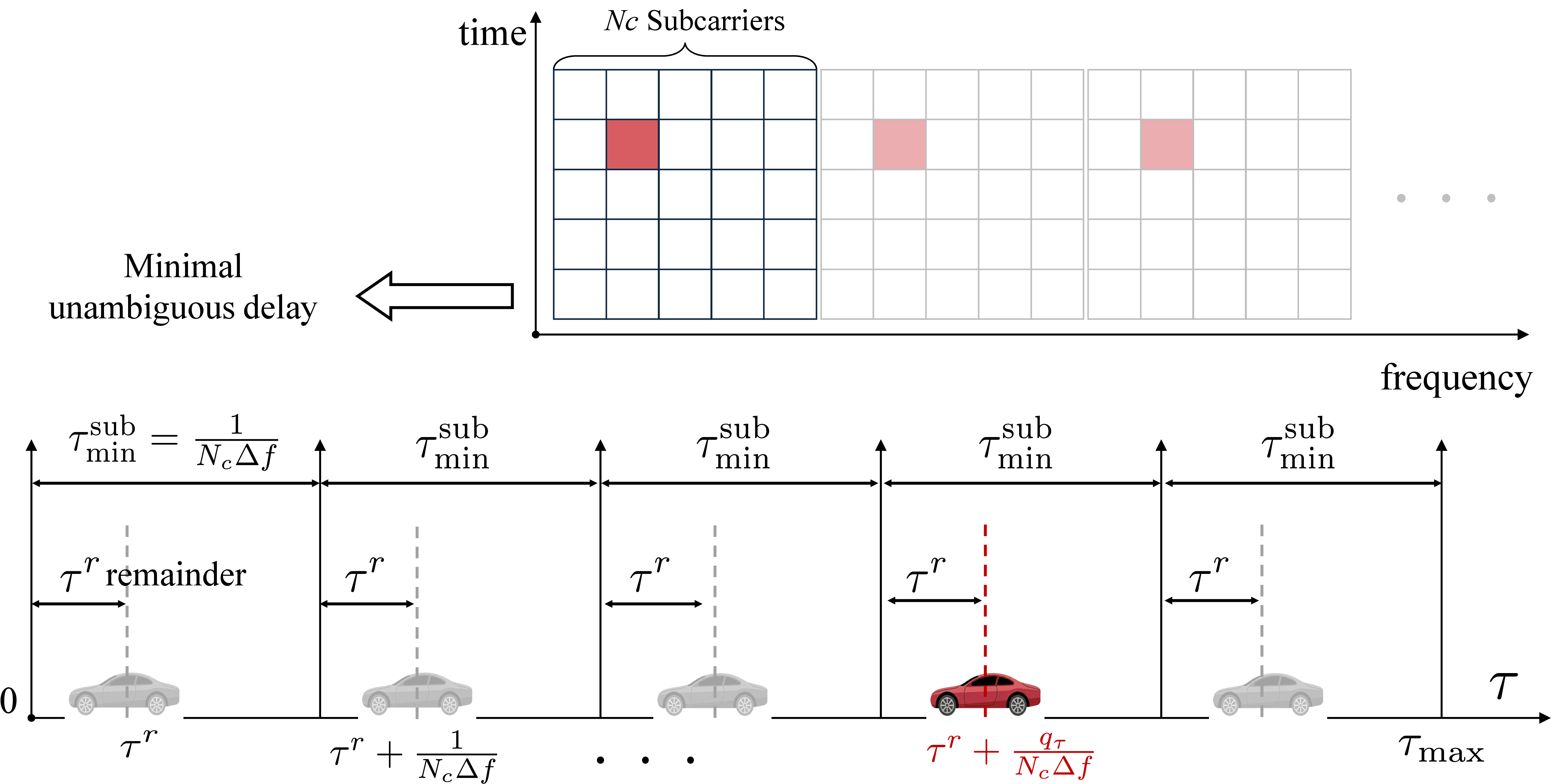}
\caption{Illustration of unambiguous sub-regions and ambiguity integers under periodic resource block sampling (Taking delay as an example).}
\label{integers}
\end{figure}

\subsubsection{Congruence Equations for Solving Ambiguity Integers}
To solve the ambiguity integers, we can leverage additional fine-grained phase shifts in the received signals introduced by the intra-block pilot patterns. In particular, the intra-block pattern can be designed to further resolve the ambiguity integers $q_\tau$ and $q_\nu$, so that the maximum unambiguous sensing region defined by $\tau_{\text{max}}$ and $ \nu_{\text{max}}$ can be achieved.

According to \eqref{phiij}, the channel response $\mathbf{\Phi}^{\mathcal P}_{\tau, \nu} \in \mathbb{C}^{N_c \times N_s}$ induced by the $(m_p, n_p)$-th intra-block pilot can be expressed as: 
\begin{equation}
\label{phaseofP}
\mathbf{\Phi}^{\mathcal P}_{\tau, \nu} (m_p, n_p) = 
 e^{j2\pi(T_s n_p\nu-\Delta f m_p \tau )}.  
\end{equation}
To enumerate the fine-grained phases brought by intra-block pilots, we calculate the pairwise phase differences by the conjugate products $\mathbf{\Phi}^{\mathcal P}_{\tau, \nu}(m_i, n_i) (\mathbf{\Phi}^{\mathcal P}_{\tau, \nu}(m_j, n_j))^*$. According to \eqref{adelayDoppleramb}, the phase of the conjugate product is given by:  
\begin{equation}
\label{delayDoppleramb}
\begin{aligned}
&e^{j2\pi (n_i-n_j) T_s \left( \nu + \frac{q_\nu}{N_s T_s} \right)}
e^{- j2\pi(m_i-m_j) \Delta f \left( \tau + \frac{q_\tau}{N_c \Delta f} \right)}\\
= &e^{j2\pi \left( \Delta n_{ij} T_s\nu - \Delta m_{ij} \Delta f\, \tau\right)}
e^{j2\pi\frac{D_{ij}}{N_c N_s} }, 
\end{aligned} 
\end{equation}
where $\{ (\Delta m_{ij}, \Delta n_{ij})=(m_i-m_j, n_i-n_j) \mid \forall i, j \in {1,...,P}, i \neq j \}$ 
denotes the index difference vectors of distinct pilots, as illustrated in Fig.~\ref{coprime}a. 
The residual phase $D_{ij}$ observed by the receiver is given by:   
\begin{equation}
D_{ij}\equiv \left( N_c \Delta n_{ij}  q_\nu -N_s \Delta m_{ij}  q_\tau \right)  \pmod {N_c N_s}.
\label{eq:Ds}
\end{equation}
To achieve unambiguous sensing, the ambiguity integer pair $(q_\tau, q_\nu)$ in the wrapped phase $D_{ij}$ must have unique solutions.

\begin{figure}[t]
\centering
\includegraphics[width=3.5in]{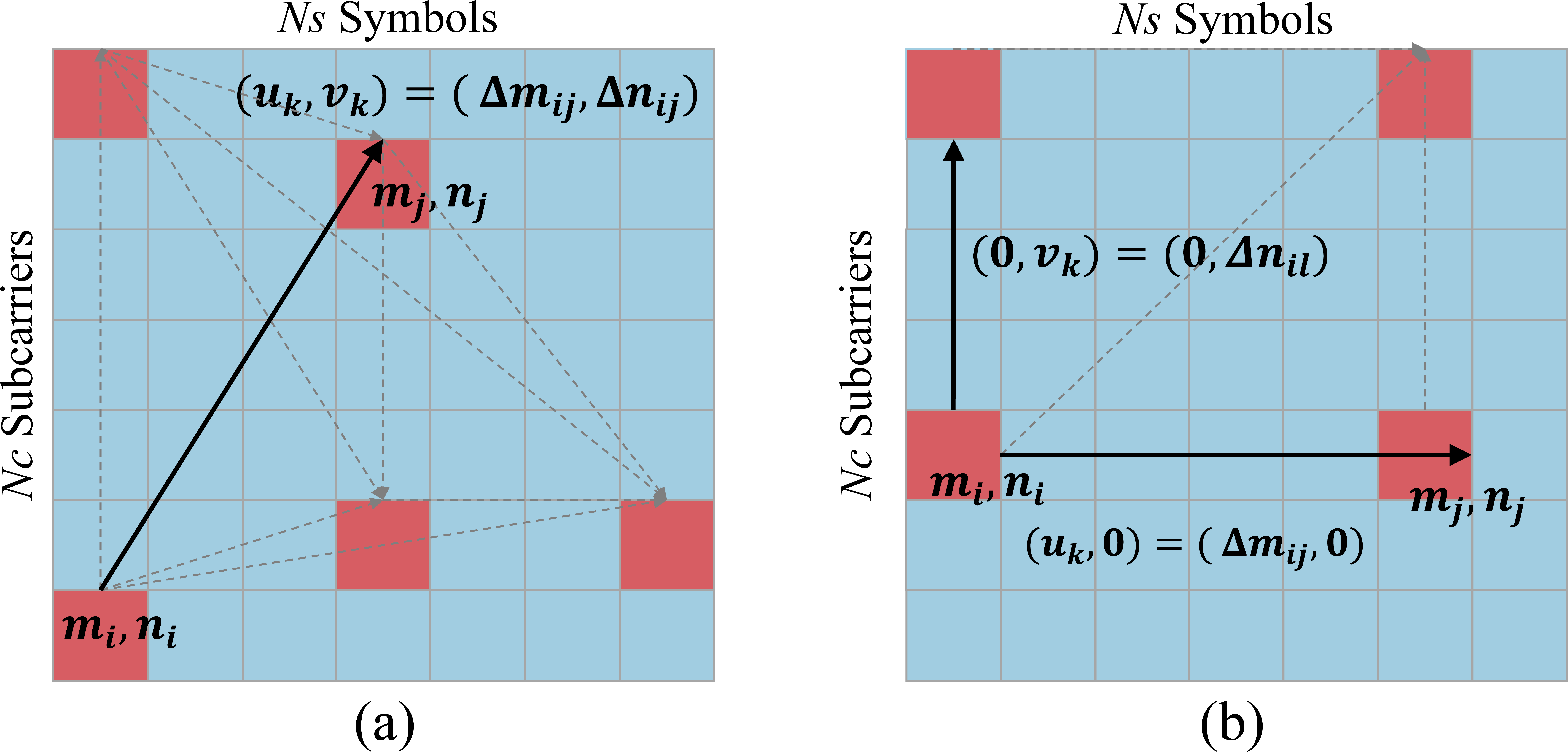}
\caption{Illustration of the pilot pattern and pilot index difference vectors of (a) 2D coprime and (b) orthogonal coprime for resolving sensing ambiguity.}
\label{coprime}
\end{figure}

For ease of illustration, we map the distinct pilot pairs $(i, j)$ into a single sequential index $k \in \{1, 2, \dots, R_d\}$, where $R_d$ denotes the total number of distinct pairwise difference vectors. Specifically, we denote the subcarrier and symbol index differences for the $k$-th pair as $(u_k, v_k)$ to replace $(\Delta m_{ij}, \Delta n_{ij})$.
Similarly, the observed phase $D_{ij}$ can be re-indexed into the vector $\mathbf{D} = [D_1, \dots, D_{R_d}]^T \in \mathbb{Z}^{R_d \times 1}$. We further define the index difference matrix $\mathbf{P} \in \mathbb{Z}^{R_d \times 2} $ as:   
\begin{equation}
\mathbf{P} = \begin{bmatrix} \mathbf{u}& \mathbf{v}\end{bmatrix}= \begin{bmatrix} u_1 & v_1 \\ u_2 & v_2 \\ \vdots & \vdots \\ u_{R_d} & v_{R_d} \end{bmatrix}.
\label{P}
\end{equation}
Consequently, equation \eqref{eq:Ds} can be rewritten as a system of congruence equations:  
\begin{equation}
\mathbf{P} \ [-N_s q_\tau, N_c q_\nu]^T \equiv \mathbf{D}\pmod {N_c N_s}.
\label{eq:pqd}
\end{equation}
The pilot index difference matrix $\mathbf{P}$ can be designed to ensure that the system of congruence equations in \eqref{eq:pqd} has a unique solution of ambiguity integers $(q_{\tau}, q_{\nu})$. 

\subsubsection{2D Coprime Condition}
To uniquely solve the ambiguity integers, $\mathbf{P}$ is required to isolate the individual integers from the ambiguity integer pairs and identify their unique values. This is essentially a weighted combination problem for difference vectors $(u_{k}, v_{k})$.
Taking delay estimation as an example, there needs to exist a set of integer weights $\boldsymbol{\alpha} = [\alpha_1, \dots, \alpha_{R_d}]^T$ such that the weighted summation of the observation yields a unique $q_\tau$ while eliminating the Doppler-dependent term $q_\nu$:
\begin{equation}\label{20}
\sum_{k=1}^{R_d} \alpha_k D_k \equiv-N_s q_\tau \pmod{N_c N_s} .
 \end{equation}
For \eqref{20} to hold for all possible ambiguity integers $q_\tau$, the elements of the index difference matrix $\mathbf{P}$ satisfying \eqref{sunau1} together with the Doppler cancellation condition, which is given by: 
\begin{equation}
\label{sunau1}
\begin{cases}
\sum_{k=1}^{R_d} \alpha_k u_k \equiv 1 \pmod{N_c}&(\text{Delay Isolation}) ,\\
\sum_{k=1}^{R_d} \alpha_k v_k \equiv 0 \pmod{N_s}&(\text{Doppler Null}). 
\end{cases} 
\end{equation}
Similarly, to recover $q_\nu$, there needs to exist weights $\boldsymbol{\beta}$ satisfying:  
\begin{equation}
\label{sunau}
\begin{cases}
\sum_{k=1}^{R_d} \beta_k u_k  \equiv 0 \pmod{N_c}&(\text{Doppler Isolation}), \\
\sum_{k=1}^{R_d} \beta_k v_k \equiv 1 \pmod{N_s}&(\text{Delay Null}).
\end{cases} 
\end{equation}
According to the generalized Bézout's identity on the 2D integer ring \cite{newman1972integral}, integer weights $\alpha$ and $\beta$ satisfying \eqref{sunau1} and \eqref{sunau} guarantee the unique recovery of $(q_\tau,q_\nu)$. The sufficiency and necessity of the existence of $\boldsymbol{\alpha}$ and $\boldsymbol{\beta}$ can be characterized by the 2D coprime condition given in the following theorem. 
\begin{theorem}[2D Coprime Condition]
\label{theorem2}
The ambiguity integers $(q_\tau, q_\nu)$ can be uniquely determined from the observed phase vector $\mathbf{D}$ if and only if there exists a subset of difference vectors in $\mathbf{P}$ $($indexed by $l,m,n \in \{1,2,\dots,R_d\})$ such that their cross-determinants and the periodicity terms have the greatest common divisor of 1: 
\begin{equation}
\label{2dcoprime}
\gcd \big( \det(\mathbf{P}_{mn}), \, u_l N_s, \, v_l N_c, \, N_c N_s \big) = 1,
\end{equation}
where $\mathbf{P}_{mn} = [u_m, v_m; u_n, v_n] $ denotes the $2 \times 2$ submatrix formed by the $m$-th and $n$-th difference vectors.
\end{theorem}
\begin{IEEEproof}
See Appendix~\ref{app:proof_thm1} for the proof.
\end{IEEEproof}
Theorem~\ref{theorem2} provides a necessary and sufficient condition for the congruence system of DD parameters to admit a unique solution. The condition essentially requires that the linear combination of pilot intervals is able to synthesize the finest symbol and subcarrier resolution grid, such that the sensing ambiguity is completely resolved.

\subsubsection{Comparison with Orthogonal Coprime}\label{compare1d2d}

As shown in Fig.~\ref{coprime}b, the orthogonal coprime design by \cite{Meiperiodiccoprime} provides a sufficient but not necessary coprime condition to resolve ambiguity.
Specifically, pilots are configured with coprime intervals across both subcarriers and symbols. Thus, the difference vectors lie on decoupled delay and Doppler axes, unnecessarily enforcing the unambiguous condition on two orthogonal dimensions. 
Formally, let $\mathcal{M} \subset \{0,1,\dots,N_c-1\}$ and $\mathcal{N} \subset \{0,1,\dots,N_s-1\}$ denote the index sets in subcarrier and symbol dimensions, respectively. The orthogonal coprime pilot pattern $\mathcal{P}_{\text{orthogonal}}$ can be expressed as the Cartesian product:  
\begin{equation}
\mathcal{M} \times \mathcal{N} = \{(m_p,n_p) \mid \forall m_p \in \mathcal{M}, n_p \in \mathcal{N}\}.
\end{equation}
Restricting pilot configuration to one axis (for example, $\forall i, v_i=0$) simplifies  the 2D coprime condition to a 1D congruence $\sum_i \alpha_i u_i = 1 \pmod{N_c}$, as the orthogonality constraint $\sum_i \alpha_i v_i = 0\pmod{N_s}$ is trivially satisfied. 

It can be observed that, compared with the 2D coprime condition in Theorem~\ref{theorem2}, the 1D coprime condition (e.g., $\sum_i \alpha_i u_i = 1 \pmod{N_c}$) is not necessary. This implies that the 2D coprime condition is more general and flexible, which leaves more freedom for pilot design to maximize sensing accuracy. 

\textbf{Remark 2:}
\textit{
The design of sensing pilots essentially involves allocating multiple comb-type pilot sequences across the entire frame. 
The difference between orthogonal and 2D coprime methods has been discussed above. We now focus on the difference between the CRT-based method and the two coprime methods. CRT-based methods \cite{liuLowOverhead,Silva2018} require different comb pilot sequences with coprime sampling intervals across the entire frame in either the time or the frequency domain, thereby satisfying the coprime condition derived from CRT to achieve unambiguous sensing. In contrast, the orthogonal coprime method and the proposed 2D coprime method design pilot patterns within the periodic resource blocks\cite{Meiperiodiccoprime}. This is equivalent to configuring multiple comb pilot sequences with identical time or frequency sampling intervals within the OFDM frame, where the relative positions among different combs must be designed to satisfy the aforementioned orthogonal or 2D coprime condition for unambiguous sensing.}

\subsection{An Insight about Sparse Sampling}\label{insight}
\begin{figure}[t]
\centering
\includegraphics[width=3.4in]{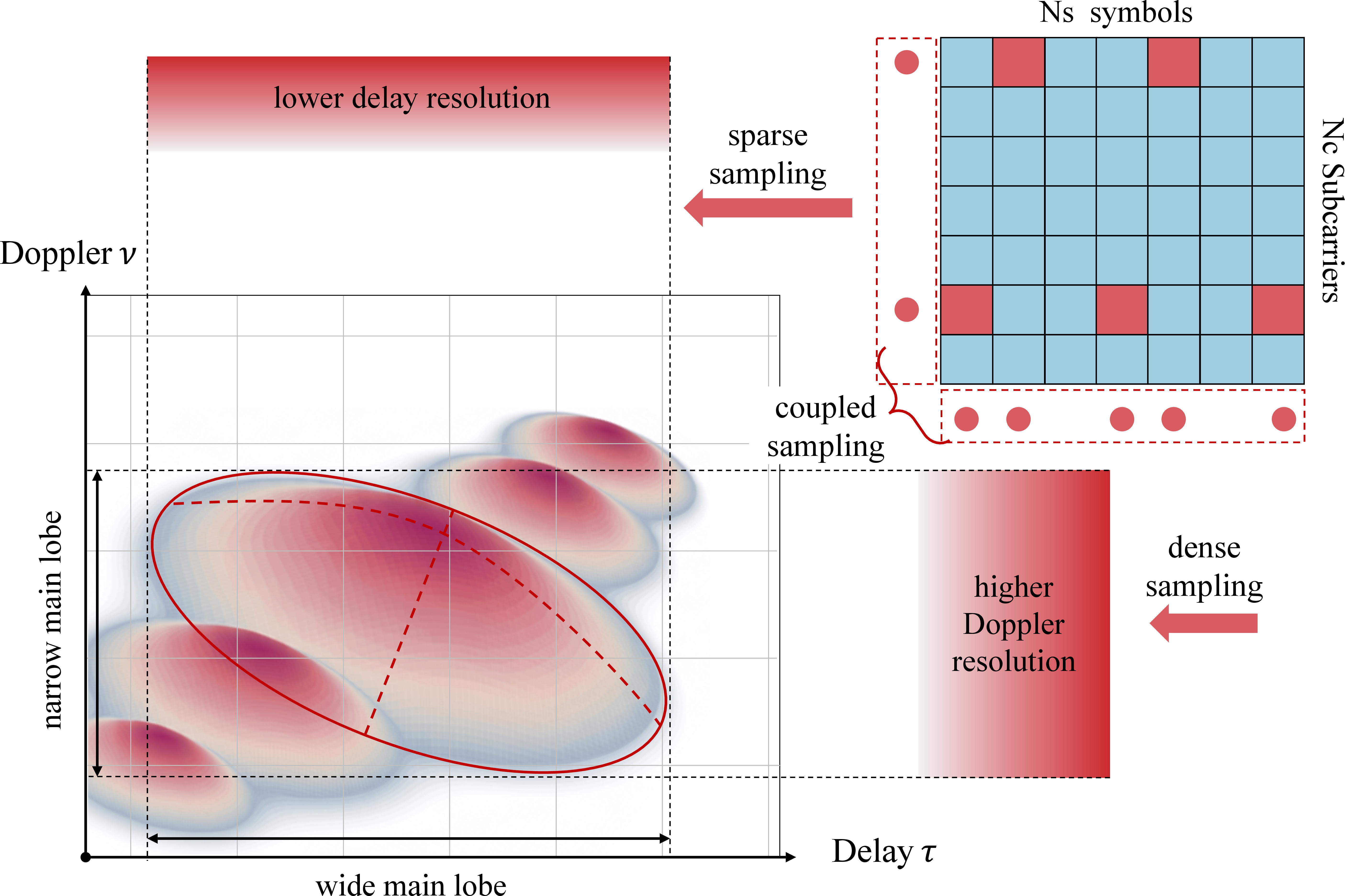}
\caption{The ambiguity function of sparse pilot in the DD domain. }
\label{sampling}
\end{figure}
Sparse sampling in DD domain inherently induces energy dispersion, which eventually causes the concerned sensing ambiguity and uncertainty. The energy dispersion in DD domain can be visualized by AF, which offers insights into the effects of sparse sampling \cite{book,singleblock}. In the following, we use AF to characterize the different pilot patterns' sensing capabilities in terms of sensing ambiguity and accuracy. 

On the one hand, sensing ambiguity can be characterized by the lobes in the AF. As shown in Fig.~\ref{sampling}, in addition to the desired main lobe corresponding to the ground truth target parameters, the multiple aliased lobes due to sparse sampling in the DD domain essentially lead to the sensing ambiguity. Since these lobes possess responses identical to the main lobe, the receiver is unable to distinguish among them.

On the other hand, the sensing accuracy can be characterized by the resolution of the AF \cite{radarAF}. For instance, denser sampling in the time domain yields a narrower main lobe of the AF along the Doppler axis, which leads to a finer resolution in the Doppler domain. 
Furthermore, the coupled sampling reshapes the lobe of the AF, as exemplified by the elliptical contour shown in Fig.~\ref{sampling}. The principal axes of the ellipse dictate the finest and coarsest sensing resolution in the DD domain. With 2D coprime sampling, the principal axis of the AF is in the delay-Doppler coupling ridge\cite{rihaczek1996principles}. Such sampling yields superior sensing accuracy for correlated DD parameters than the orthogonal coprime solution \cite{dataset}, whose principal axes of the AF are along the delay and Doppler axes.

\subsection{Sensing Accuracy Evaluation}
As shown above, the 2D coprime condition guarantees the unique main lobe in AF and leaves more freedom to optimize sensing accuracy than existing solutions.  
However, optimizing the sensing accuracy $\mathcal{U}(\mathcal{P})$ (i.e., the SMI within the unambiguous sensing region) in \eqref{problemformulation} is intractable due to the lack of a closed-form expression. Therefore, we resort to the BCRB as a tractable surrogate for SMI. An AGIR algorithm is then proposed to optimize this surrogate under the 2D coprime condition.

We first establish the relationship between BCRB and SMI.  Assuming that $\boldsymbol{\theta}$ follows the joint real Gaussian distribution with covariance matrix $\mathbf{\Sigma}_\theta\in \mathbb{R}^{2 \times 2}$. The relation between the BCRB matrix $\Psi_\text{BCRB}(\mathcal{P})\in \mathbb{R}^{2 \times 2}$ of the pilot pattern $\mathcal{P}$ and the SMI between $\boldsymbol{\theta}$ and $\mathbf{Y}_\text{RS}$ can be given by \cite{xie2024}: 
\begin{equation}\label{lowerbound}
I(\mathbf{Y}_\text{RS}; \boldsymbol{\theta}) \ge \frac{1}{2}(\log \det(\mathbf{\Sigma}_\theta) - \log \det (\Psi_\text{BCRB}(\mathcal{P}))), 
\end{equation}
which provides an asymptotically tight lower bound as the SNR increases\cite{BayesianFIM}. Consequently, the log-determinant form of the BCRB can serve as an effective surrogate for $\mathcal{U}(\mathcal{P})$.

Next, we calculate the BCRB. Given the prior $\mathbf{\Sigma}_\theta$, the posterior BCRB $\Psi_\text{BCRB}(\mathcal{P})$ in \eqref{lowerbound}, which provides a lower bound on the estimation error for $\boldsymbol{\theta}$, is defined as the inverse of the Bayesian Fisher Information Matrix (FIM) $\mathbf{J}\in \mathbb{R}^{2 \times 2}$. In particular, the Bayesian FIM $\mathbf{J}$ for $\boldsymbol{\theta}$ is given by\cite{vantrees2007bayesian}:  
\begin{equation}\label{J}
\begin{aligned}
 &\mathbf{J}  \triangleq \mathbb{E}_{\boldsymbol{\theta}} \left(
\frac{\partial \log p(\boldsymbol{\theta})}{\partial \boldsymbol{\theta}}
\frac{\partial \log p(\boldsymbol{\theta})}{\partial \boldsymbol{\theta}^{T}}
\right) 
\\  + &\mathbb{E}_{\mathbf{Y}_\text{RS},\boldsymbol{\theta}} \left(
\left.
\frac{\partial \log p(\mathbf{Y}_\text{RS}|\mathcal{P},\boldsymbol{\theta})}{\partial \boldsymbol{\theta}}
\frac{\partial \log p(\mathbf{Y}_\text{RS}|\mathcal{P},\boldsymbol{\theta})}{\partial \boldsymbol{\theta}^{T}}
\right| \mathcal{P}
\right). 
\end{aligned}
\end{equation}
The BCRB can thus be calculated as:   
\begin{equation}\label{BCRB}
\Psi_\text{BCRB}(\mathcal{P}) =\mathbf{J}^{-1} = \left( \mathbf{\Sigma}_\theta^{-1} + \mathbf J_{\mathbf{Y}_\text{RS}|\mathcal{P}}(\mathcal{P}) \right)^{-1}, 
\end{equation}
where the inverse of covariance matrix $\mathbf{\Sigma}_\theta$ is equal to the first term in \eqref{J}, and $\mathbf J_{\mathbf{Y}_\text{RS}|\mathcal{P}}(\mathcal{P})\in \mathbb{R}^{2 \times 2}$ is the likelihood FIM for observation $\mathbf{Y}_\text{RS}$, corresponding to the second term in \eqref{J}.

To concisely calculate the $\mathbf J_{\mathbf{Y}_\text{RS}|\mathcal{P}}(\mathcal{P})$, only the received pilot observations are retained from $\mathbf{Y}_\text{RS}$. 
According to \eqref{eq:rx_signal_re1}, the effective observation can be vectorized as: 
\begin{equation}
\mathbf y(\mathcal P)
=
\xi\boldsymbol{\mu}(\mathcal P)+\mathbf w
\in\mathbb C^{P\times1},
\end{equation}
where the propagation coefficient is modeled as a zero-mean circularly symmetric complex Gaussian random variable  
$\xi\sim\mathcal{CN}(0,\sigma_\xi^2)$ that is independent of $\boldsymbol\theta$, the $p$-th entry of $\boldsymbol{\mu}(\mathcal P)$ is given by $\mathbf{\Phi}_{\tau, \nu} (m_p, n_p)\mathbf{X}_\text{RS} (m_p,n_p)$, and $\mathbf w$ is the additive noise following $\mathcal{CN}(0,\sigma^2 \mathbf I_P)$. 
The conditional distribution $\mathbf y(\mathcal P)\,|\,\boldsymbol\theta$ follows a complex Gaussian distribution 
$
\mathcal{CN}
\!\left(
\mathbf0,
\mathbf R
\right) $ with covariance matrix: 
\begin{equation}
\mathbf R
=
\sigma_\xi^2
\boldsymbol\mu(\mathcal P)
\boldsymbol\mu^H(\mathcal P)
+
\sigma^2\mathbf I, 
\end{equation}
which depends on the given $\boldsymbol\theta$ only through
$\boldsymbol\mu(\mathcal P)$.

Since the marginalized observation is zero mean, the mean-dependent term in the Slepian--Bangs formula vanishes\cite{kay}, and the equivalent likelihood FIM is thus given by:
\begin{equation}
\label{slep}
[\mathbf J_{\rm eq}(\mathcal P)]_{i,j}
=
{\rm tr}
\!\left(
\mathbf R^{-1}
\frac{\partial\mathbf R}{\partial\theta_i}
\mathbf R^{-1}
\frac{\partial\mathbf R}{\partial\theta_j}
\right),
\end{equation}
where $\theta_1=\tau$, $\theta_2=\nu$, and
\begin{equation}
\label{grad}
\frac{\partial\mathbf R}{\partial\theta_i}
=
\sigma_\xi^2
\left(
\mathbf d_i
\boldsymbol\mu^H
+
\boldsymbol\mu
\mathbf d_i^H
\right). 
\end{equation} 
Here, $\mathbf d_i$ is the $i$-th column of the Jacobian matrix $\mathbf D
=
\left[
\mathbf d_1,\,
\mathbf d_2
\right]
=
\left[
\frac{\partial\boldsymbol\mu}{\partial\tau},
\frac{\partial\boldsymbol\mu}{\partial\nu}
\right]$, whose columns are given by: 
\begin{equation}
    \begin{cases}
&\mathbf{d}_1
=
-j2\pi\Delta f\,
\operatorname{diag}
\left(
\left\{
m_p+b_fN_c
\right\}_{p=1}^{P}
\right)
\boldsymbol{\mu},
\\
&\mathbf{d}_2
=
j2\pi T_s\,
\operatorname{diag}
\left(
\left\{
n_p+b_tN_s
\right\}_{p=1}^{P}
\right)
\boldsymbol{\mu}.
    \end{cases}
\end{equation}
Using the Woodbury matrix identity, the inverse covariance matrix is obtained as:
\begin{equation}\label{inv}
\mathbf R^{-1}
=
\frac1{\sigma^2}
\left(
\mathbf I
-
\frac{\sigma_\xi^2}
{\sigma^2+\sigma_\xi^2P}
\boldsymbol\mu
\boldsymbol\mu^H
\right). 
\end{equation}
After substituting \eqref{grad} and \eqref{inv} into \eqref{slep} and simplifying the trace terms, the $(i,j)$-th entry of $\mathbf J_{\rm eq}(\mathcal P)$ is given by: 
\begin{equation}
[\mathbf J_{\rm eq}]_{i,j}
=
\frac{2\sigma_\xi^4P}
{\sigma^2(\sigma^2+\sigma_\xi^2P)}
{\rm Re}
\left\{
\mathbf d_i^H\mathbf d_j
-
\frac{
\mathbf d_i^H
\boldsymbol\mu
\boldsymbol\mu^H
\mathbf d_j
}
{\|\boldsymbol\mu\|^2}
\right\}, 
\label{eq:Jij_new}
\end{equation}
with $\mathbf{d}_{1}^H \mathbf{d}_{1} = 4\pi^2 \Delta f^2 \sum_{p=1}^{P} (m_p + b_f N_c)^2$ and $\boldsymbol{\mu}(\mathcal{P})^H \mathbf{d}_{1}=-j 2\pi \Delta f \sum_{p=1}^{P} (m_p + b_f N_c)$. The remaining inner products can be obtained in the same manner.

The inner products in \eqref{eq:Jij_new} depend only on the centered second-order statistics of the intra-block pilot indices and is independent of $\boldsymbol{\theta}$. Hence, averaging it over the prior distribution does not change the result.

\subsection{Sensing Accuracy Optimization}

Finally, the optimization problem in \eqref{problemformulation} can be reformulated as the BCRB optimization problem subject to pilot overhead and 2D coprime constraints:  
\begin{equation}\label{eq:optimizationproblem}
\begin{aligned}
    \max_{\mathcal{P}}\quad & \mathcal{\bar{U}}(\mathcal{P}) = \log \det \left( \mathbf{\Sigma}_\theta^{-1} + \mathbf J_\text{eq}(\mathcal{P}) \right), \\
    \text{s.t.} & \quad \mathcal{P} \in \Omega , \\
    & \quad P = N_{\text{RS}}. 
\end{aligned}
\end{equation}
Due to the coprime condition and the integer constraints on block scales and guard intervals\cite{separateanalysis}, the optimization problem \eqref{eq:optimizationproblem} constitutes a non-convex integer programming problem.

Inspired by the continuous manifold projection\cite{Bazzi2025} and iterative refinement strategies\cite{Stoica2009}, we propose an AGIR algorithm to find the optimal pilot pattern. The details of the AGIR algorithm are given in Algorithm~\ref{alg:agir}. By relaxing the discrete constraints into the continuous space, the algorithm calculates the continuous gradient of the objective. This allows the optimizer to capture the underlying trend of the manifold and rapidly guide the pilot pattern toward the optimal neighborhood.

\begin{algorithm}[tbp] 
\caption{Adam-Guided Iterative Refinement (AGIR)}
\label{alg:agir}
\begin{algorithmic}[0]
\Require Grid size $N_c, N_s$; Pilot number $P$; Prior covariance $\boldsymbol{\Sigma}_{\theta}$; Max epochs $T_\text{epoch}$; Gradient steps $T_\text{grad}$; Adam hyperparameters $\beta_1, \beta_2, \epsilon, \eta$.
\Statex \textbf{Initialization with valid set $\Omega_{\text{valid}}$:}
\State Randomly initial pattern $\mathcal{P}^{(0)} = \{(m_p, n_p)\}_{p=1}^{P}$ with distinct values that satisfy the 2D coprime $\gcd(\mathcal{P}^{(0)}) = 1$ and guard interval constraint $|\Delta  m_{ij}| +|\Delta  n_{ij}| > D_\text{gurad}, \forall i \neq j$.
\State Initialize Adam moments: $\mathbf{m} \gets \mathbf{0}, \mathbf{v} \gets \mathbf{0}, t \gets 0$.
    \For{epoch $k = 1$ \textbf{to} $T_\text{epoch}$}
        \Statex \textit{Phase 1:  Adam Iterative Exploration}
        \State Discrete to continuous pattern: $\mathcal{Q}^{(k)} \gets \mathcal{P}^{(k-1)}$.
        \For{$j = 1$ \textbf{to} $T_\text{grad}$}
            \State $t \gets t + 1$
            \State Compute objective function gradient: 
                   $\mathbf{g}_t \gets \frac{\partial \mathcal{\bar{U}}(\mathcal{\mathcal{Q}})}{\partial \mathcal{Q}}$.
            \State Update first moment: 
                   $\mathbf{m} \gets \beta_1 \mathbf{m} + (1 - \beta_1) \mathbf{g}_t$.
            \State Update second moment: 
                   $\mathbf{v} \gets \beta_2 \mathbf{v} + (1 - \beta_2) \mathbf{g}_t^2$.
            \State Compute: 
                   $\hat{\mathbf{m}} \gets \mathbf{m} / (1 - \beta_1^t)$, $\hat{\mathbf{v}} \gets \mathbf{v} / (1 - \beta_2^t)$.
            \State  Adaptive step-size update: 
                   $\mathcal{Q}^{(k)} \gets \mathcal{Q}^{(k)} + \eta \cdot \frac{\hat{\mathbf{m}}}{\sqrt{\hat{\mathbf{v}}} + \epsilon}$.
            \State Project: 
                   $\mathcal{Q}^{(k)} \gets \operatorname{Clip}(\mathcal{Q}^{(k)}, [0, N_c-1] \times [0, N_s-1])$.
        \EndFor
        \Statex \textit{Phase 2: Nearest Projection and Discrete Refinement}
        \State Project continuous $\mathcal{Q}^{(k)}$ to the nearest distinct integers satisfying constrains: $\mathcal{P}_{\text{temp}} \gets \arg \min_{\mathcal{P} \in \Omega_{\text{valid}}} \|\mathcal{P} - \mathcal{Q}^{(k)}\|^2$.
        \For{$p = 1$ \textbf{to} $P$}
            \State Fix other pilots: $\mathcal{P}_{\text{temp}} \setminus \{(m_p, n_p)\}$.
            \State Search optimal $(m_p, n_p)^\star$ over local integer neighborhood to maximize the objective function within valid set.
            \State Update $(m_p, n_p) \gets (m_p, n_p)^\star$.
        \EndFor
        \State Set $\mathcal{P}^{(k)} \gets \mathcal{P}_{\text{temp}}$.
    \EndFor
    \State \Return Optimal coprime pilot pattern $\mathcal{P}^* \gets \mathcal{P}^{(k)}$.
\end{algorithmic}
\end{algorithm}

Fig.~\ref{AGIR} illustrates the zigzag iteration process of AGIR, reflecting the alternating nature of the proposed method. Specifically, in each epoch, the algorithm temporarily relaxes the integer and coprime constraints to continuously explore along the adaptive gradient ascent direction (red trajectory). After exploration, the solution is discretized and projected back to the nearest lattices $\Omega_{\text{valid}}$ that satisfy the coprime and guard interval constraints, causing a fallback from its continuous optimal solution to a suboptimal discrete point. Using the projected point as a seed, AGIR performs discrete refinement using a restricted descent, recovering performance losses by fine-tuning pilot indices within local integer neighborhoods. This mechanism ensures that the discrete search is not random but warm-started by gradient information, which enables the algorithm to converge to a high-quality solution.

Next, we compare the complexity of the proposed AGIR algorithm with that of Exhaustive Search (ES) and Greedy Search (GS). ES requires an exponential complexity of $\mathcal{O}((N_c N_s)^{P})$, which is intractable for large resource blocks. In contrast, GS operates with a highly efficient linear complexity of $\mathcal{O}(P \cdot (N_c N_s))$ but often falls into local optima without considering the joint coupling among pilots. The proposed AGIR alternates between continuous Adam exploration and discrete refinement for a few epochs $T_\text{epoch}$. Its overall complexity is around $\mathcal{O}(T_\text{epoch} \cdot P \cdot (N_c N_s))$. Therefore, AGIR achieves the same linear complexity as GS, while effectively avoiding suboptimal traps through continuous gradient guidance.

\begin{figure}[t] 
\centering
\includegraphics[width=3.3in]{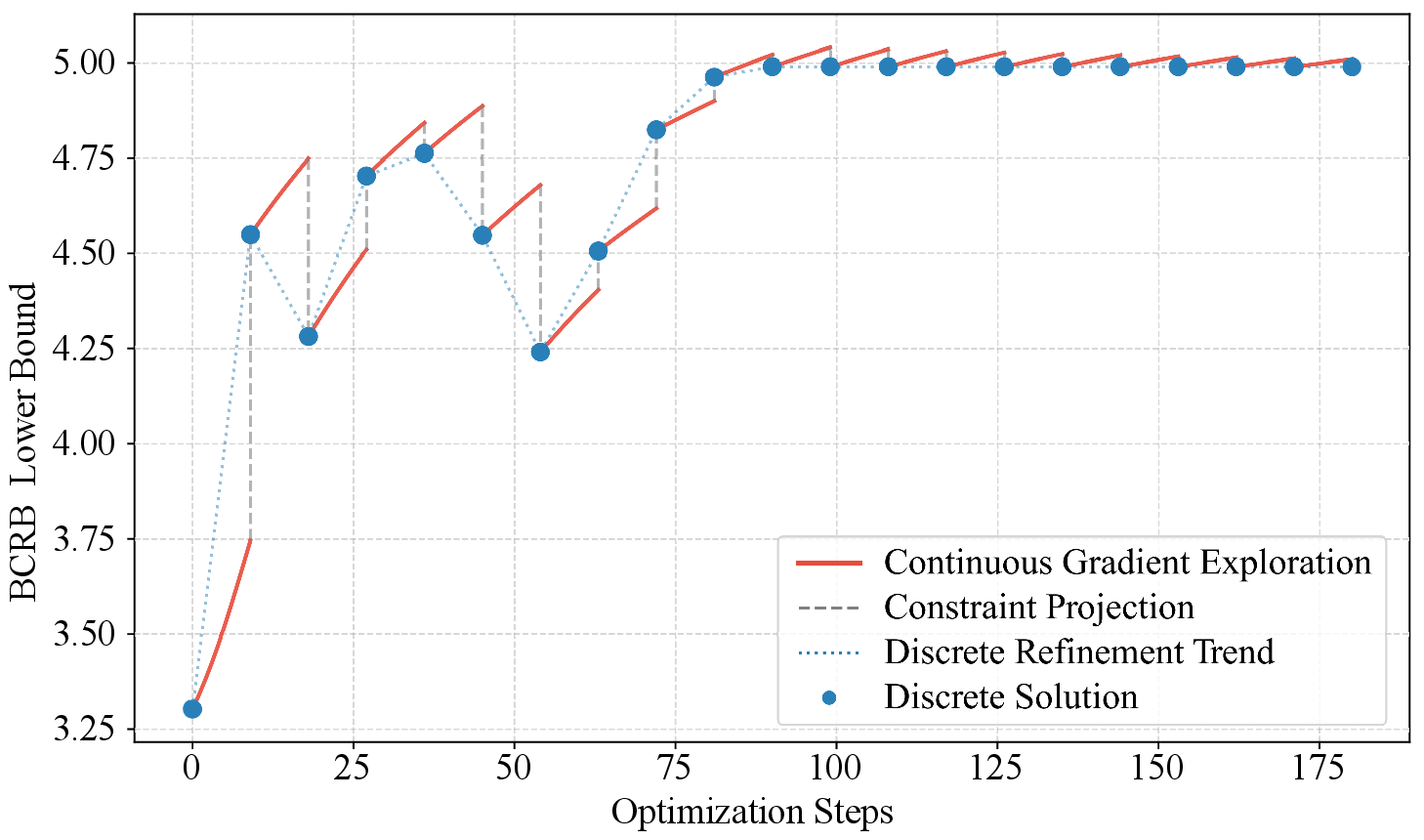}
\caption{Illustration of the zig-zag climbing behavior of the AGIR algorithm during iterative optimization.}
\label{AGIR}
\end{figure}

\section{Simulation results}

\begin{figure*}[!tbp] 
\centering
\includegraphics[width=7in]{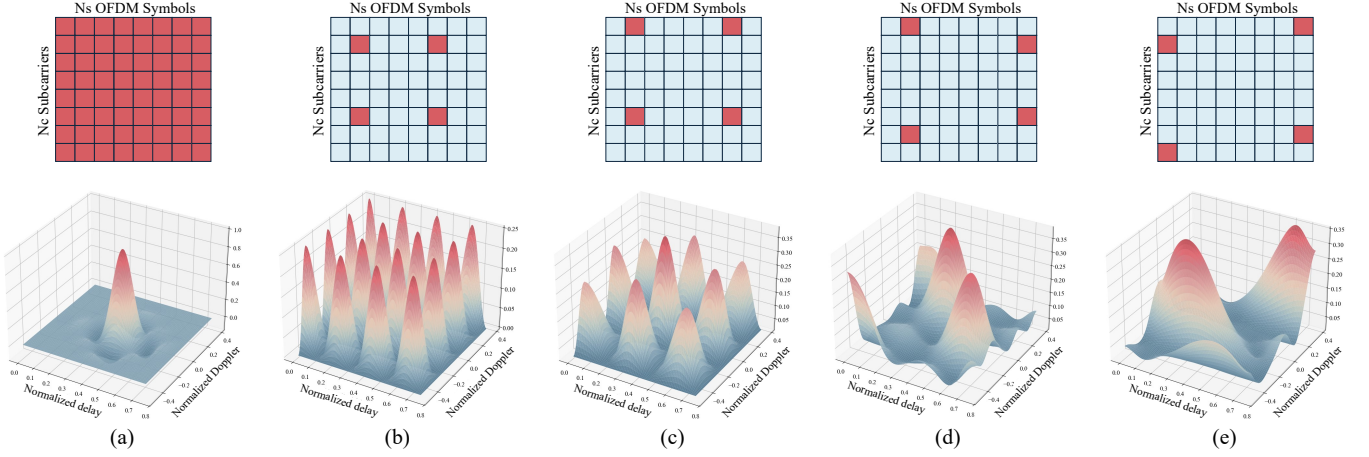}
\caption{Pilot patterns and corresponding AF in DD domain: 
(a) Full Sampling, (b) Uniform Sampling, (c) Orthogonal Coprime, (d) 2D Coprime and (e) 2D Non-Coprime. Pattern (b), (c), (d), (e) employ the same number of pilots $N_\text{RS}=4$.}\label{figPSF}
\end{figure*}

In this section, numerical results are presented to demonstrate the superiority of the proposed framework over existing designs. The advantages mainly arise from the joint optimization of the unambiguous sensing region and sensing accuracy.

\subsection{Comparison of AFs for different patterns}

We first utilize AF to illustrate the sensing performance of five different pilot patterns. The specific patterns and the corresponding AFs are illustrated in Fig.~\ref{figPSF}.


(a) Full Sampling pattern utilizes all resources for sensing. It achieves the ideal AF response with a sharp single main lobe for high DD resolution without side lobes, at the cost of very high overhead.

(b) Uniform Sampling pattern adopts pilots with uniform intervals along both subcarriers and symbols. While this pattern maintains low overhead, the sub-Nyquist sampling induces severe aliasing and restricts the maximum unambiguous regions. This manifests as multiple identical lobes in the AF, making it impossible for the receiver to identify the true target parameters.

(c) Orthogonal Coprime pattern employs coprime intervals along both subcarriers and symbols. As depicted in Fig.~\ref{figPSF}c, as the side lobes are lower than the main lobe in the AF, the receiver can unambiguously distinguish the true target parameters. Moreover, the principal axes of these lobes are parallel to the delay and Doppler axes, yielding the isotropic sensing resolutions of the orthogonal pilot pattern. Consequently, this pattern is only suitable for targets with independent DD parameters.

(d) 2D Coprime pattern utilizes the pilot indices satisfying the 2D Coprime condition given in Theorem~\ref{theorem2}. Similar to Orthogonal Coprime, it ensures unambiguous sensing due to the difference between the main and side lobes. More importantly, the coupled time-frequency sampling enhances the resolution of the AF along the coupling ridge of delay and Doppler, which enables higher sensing accuracy for DD parameters with joint distributions.

(e) 2D Non-Coprime pattern also employs coupled sampling like 2D Coprime pattern and exhibits a superior sensing resolution along the coupling ridge of delay and Doppler. However, this pattern does not satisfy the 2D Coprime condition and causes sensing ambiguity, which can be explained by the two identical lobes in the AF.

\subsection{SMI Performance Comparison}

In this subsection, we evaluate the SMI achieved by different pilot patterns with respect to SNR, pilot overhead, and resource block size. 
The baseline patterns include Full Sampling, Uniform Sampling, Orthogonal Coprime, 2D Coprime, and D-optimal (which optimizes the log-determinant BCRB surrogate for maximum sensing accuracy without satisfying the coprime condition). The pilot overhead is set to $N_{\text{RS}}=4$ within a resource block of size $8 \times 8$. 

\textbf{Remark 3:}
\textit{
The SMI in \eqref{smi1} lacks a closed-form expression. To address this, we propose a numerical method to evaluate SMI using discrete Monte Carlo simulation. The detailed numerical procedure is given in Appendix~\ref{app:smi_computation}. 
}

\subsubsection{SMI versus SNR}

\begin{figure}[tbp]
\centering
\subfloat[]{\includegraphics[width=1.714in]{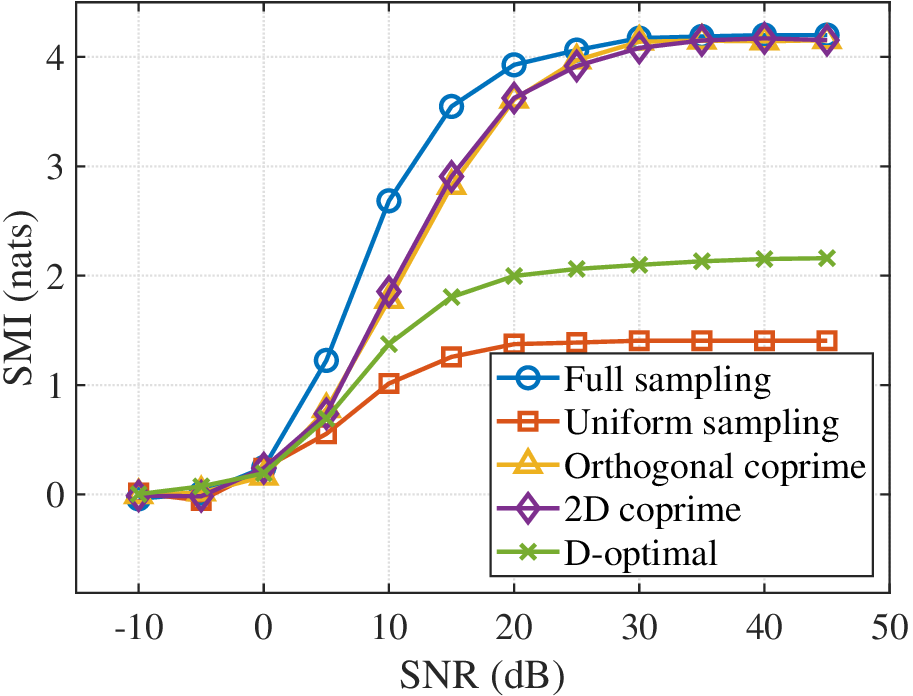}\label{fig:sec1}}
\hspace{0.005in}
\subfloat[]{\includegraphics[width=1.714in]{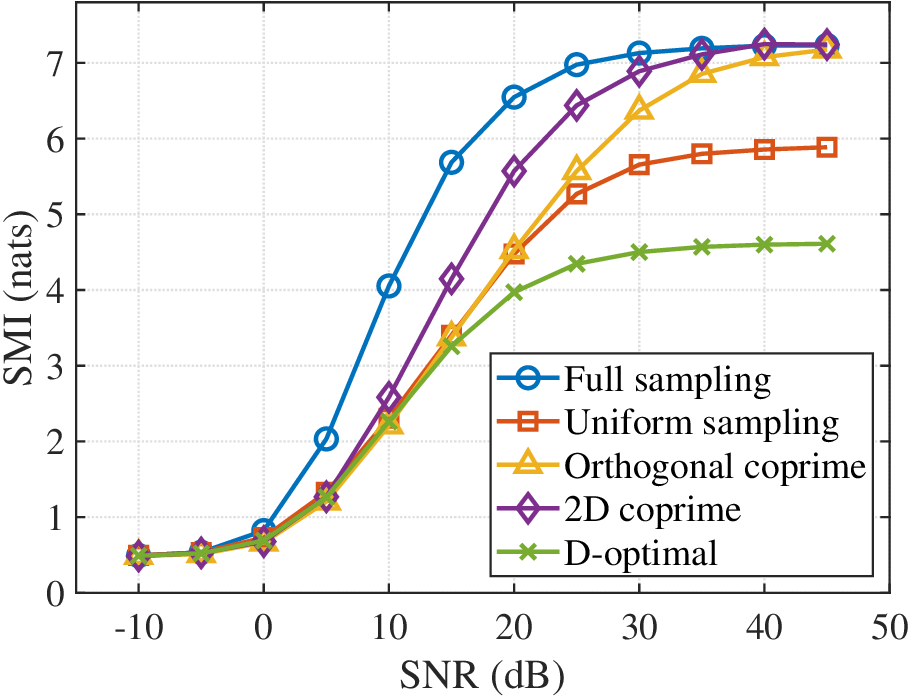}\label{fig:sec2}}
\caption{SMI versus SNR with different pilot patterns for DD parameters with (a) independent and uniform distribution and (b) joint Gaussian distribution.}
\end{figure}

We first consider the case where the DD parameters are independent and uniformly distributed within the maximum unambiguous region, and the results are shown in Fig.~\ref{fig:sec1}. It can be observed that Full Sampling achieves the optimal sensing region and accuracy, thereby establishing the upper bound of the SMI. The Orthogonal Coprime and 2D Coprime patterns achieve identical SMI values, due to the independent distributions in DD domain. At moderate SNR, the SMI of these two coprime patterns is lower than that of Full Sampling due to worse sensing accuracy. However, in the high-SNR regime, they are able to achieve the optimal SMI. 
Conversely, the D-optimal and Uniform Sampling patterns suffer severe losses. In particular, the SMI of the Uniform Sampling exhibits a loss of $\ln(4 \times 4) \approx 2.77$ nats in SMI. This degradation is because  Uniform Sampling shrinks the unambiguous sensing region to $1/4$ of the maximum sensing region for both delay and Doppler parameters.

Fig.~\ref{fig:sec2} illustrates the results where the DD parameters follow a joint Gaussian distribution. It can be observed that the proposed 2D Coprime pattern outperforms the Orthogonal Coprime pattern. This superiority is due to its ability to leverage coupled time-frequency sampling to obtain higher sensing accuracy for DD parameters with a joint distribution, as discussed in Sec.~\ref{insight}. On the other hand, although the performance gaps between different patterns reduce with joint distributions between the DD parameters, the SMI losses of Uniform Sampling and D-optimal patterns persist due to the unresolved sensing ambiguity.

The experiment results also reveal the underlying relationship between the sensing region and accuracy. In particular, pilot patterns that violate the coprime constraint (Uniform Sampling and the D-optimal pattern) consistently suffer from a severe SMI loss. This phenomenon validates the analysis in Sec.~\ref{secproblemformulation}, i.e., ambiguity is the dominant factor in maximizing SMI, as it sets the upper limit of SMI. 

\textbf{Remark 4:}
\textit{
We note that there exists a trade-off between the unambiguous sensing region and sensing accuracy. In a low SNR region, sacrificing the sensing region to slightly enhance sensing accuracy might yield a marginally higher SMI. However, under such severe noise conditions, the extremely high uncertainty also pushes the SMI to be very small, rendering the potential improvement in SMI marginal. Therefore, we ignore this case and maximize the SMI by first resolving the sensing ambiguity and then optimizing the sensing accuracy.
}

\begin{figure}[t] 
\centering
\includegraphics[width=3.2in]{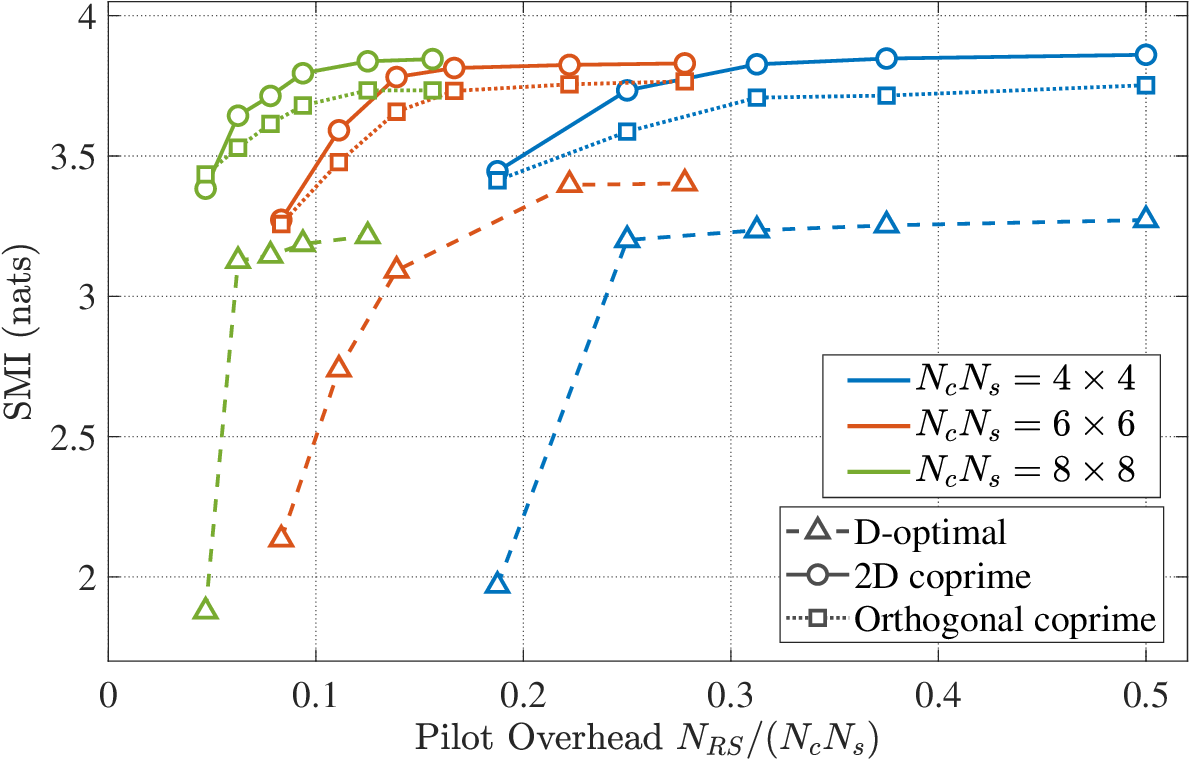}
\caption{SMI versus pilot overhead with different pilot patterns under varying resource block scales.}
\label{fig:SMI_Overhead}
\end{figure}

\subsubsection{SMI versus Sensing Overhead}

Fig.~\ref{fig:SMI_Overhead} illustrates the SMI as a function of the normalized pilot overhead, defined as $N_{\text{RS}}/(N_c N_s)$, under various resource block scales ($4 \times 4$, $6 \times 6$, and $8 \times 8$). We can observe that the SMI of all patterns increases with the pilot overhead and approaches the upper limit. The proposed 2D Coprime pattern consistently achieves the highest SMI across all overheads. The Orthogonal Coprime pattern remains suboptimal because its rigid structural requirement leaves little freedom to optimize sensing accuracy. Furthermore, the D-optimal pattern, which maximizes sensing accuracy without satisfying the coprime condition, exhibits a much lower SMI due to its failure to resolve sensing ambiguity. 

Notably, patterns with larger resource blocks require a considerably lower pilot overhead to obtain the maximum SMI. This occurs because larger $N_c$ and $N_s$ significantly enlarge the design space. However, unlocking this performance gain comes at the expense of increased computational complexity.

\subsection{Sensing Accuracy Comparison}

In this subsection, we compare the sensing accuracy achieved by different pilot patterns.

\subsubsection{Sensing Accuracy versus SNR}

Fig.~\ref{fig:mse} shows the MSE for delay and Doppler estimation by different pilot patterns versus SNR. The DD parameters are estimated using the maximum likelihood estimator. Note that, for the pilot pattern that cannot guarantee unambiguous sensing, e.g., D-optimal, we only evaluate the local MSE within the unambiguous region. 

It can be observed that the proposed 2D Coprime pattern achieves significantly lower MSE than Orthogonal Coprime, demonstrating the advantage of time-frequency sampling in improving sensing accuracy.
More importantly, the MSE of the 2D Coprime pattern closely approaches that of the D-optimal pattern, which achieves the optimal sensing accuracy (but with sensing ambiguity). This indicates that the 2D Coprime condition not only ensures the maximum unambiguous sensing region but also provides sufficient structural flexibility for pilot design to achieve near-optimal sensing accuracy.

\begin{figure}[tbp]
\centering
\subfloat[]{\includegraphics[width=1.714in]{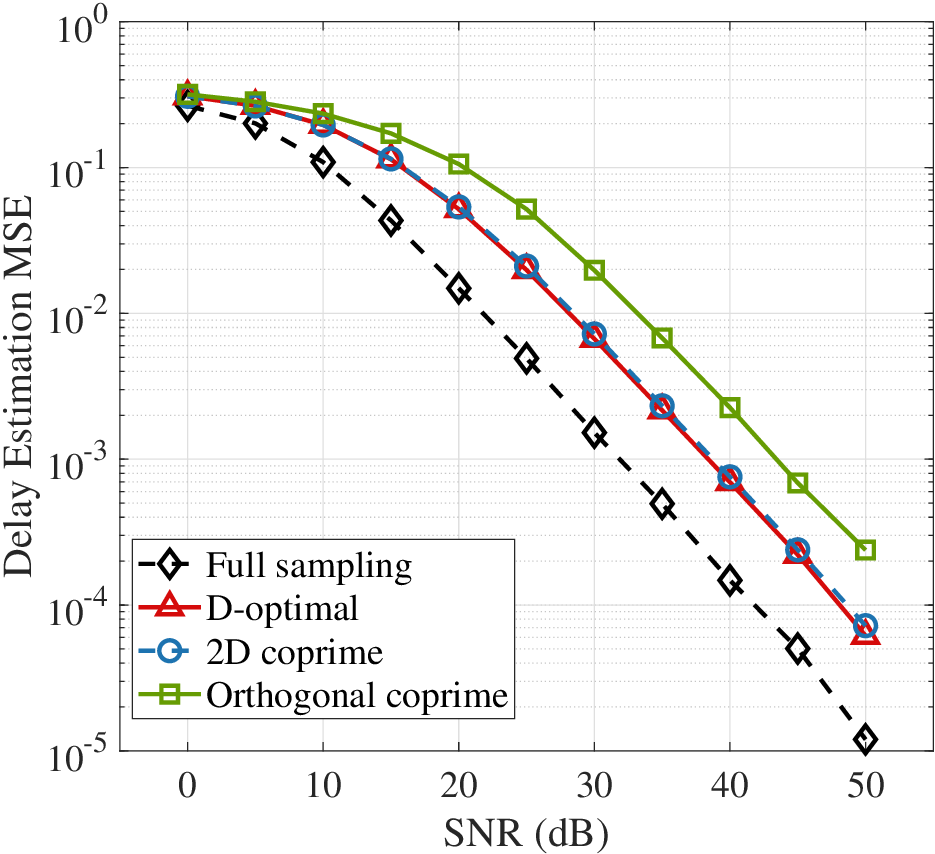}\label{fig:tau}}
\hspace{0.005in}
\subfloat[]{\includegraphics[width=1.714in]{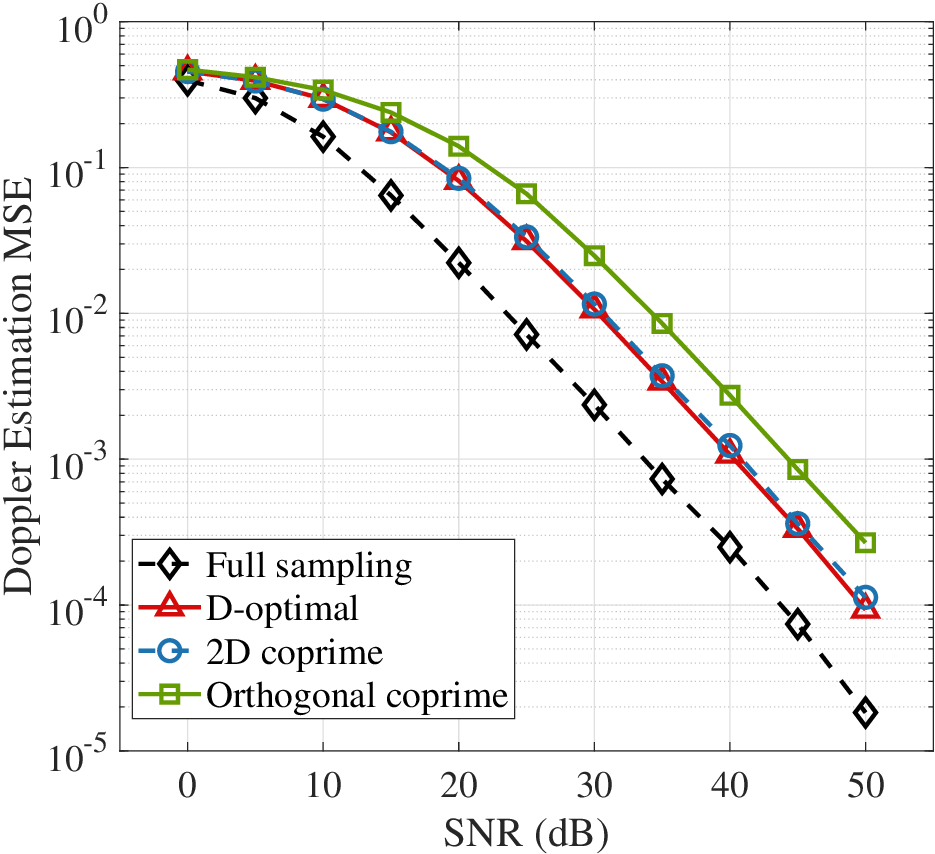}\label{fig:nu}}
\caption{The normalized MSE for (a) delay (b) Doppler parameter estimation as a function of SNR with different pilot patterns.}
\label{fig:mse}
\end{figure}

\subsubsection{The Structural Flexibility of 2D Coprime}

To illustrate the structural flexibility of the proposed 2D Coprime condition, we evaluate the adaptability of the pilot patterns with varying DD parameter distributions. For that purpose, we first define the joint distribution of DD parameters $\boldsymbol{\theta}$. Specifically, its covariance $\boldsymbol{\Sigma}_\theta$ can be characterized as an uncertainty ellipse, which is defined by the condition number $\kappa$ and the rotation angle $\phi$. The condition number $\kappa$ is defined as the ratio between the largest eigenvalue and the smallest eigenvalue, representing the relative uncertainty along the two principal axis directions of the ellipse in DD domain. Meanwhile, $\phi$ dictates the orientation of the principal eigenvectors that reflects the correlation of DD parameters along the two principal axis directions. By sweeping the condition number $\kappa$ and the rotation angle $\phi$, we demonstrate the flexibility of the proposed 2D Coprime condition in adaptively optimizing the pilot's sensing accuracy for DD parameters with joint distributions.

Fig.~\ref{fig:Ellipse} visualizes the alignment between the equivalent FIM $\mathbf J_{\text{eq}}$ and the covariance $\boldsymbol{\Sigma}_\theta$. Specifically, we gradually increase the condition number $\kappa$ from $1.0$ (isotropic) to $64.0$ (highly anisotropic). Due to its rigid structure, the Orthogonal Coprime pattern cannot adapt to the correlated DD parameters. Consequently, its FIM remains isotropic, resulting in a severe failure to accurately sense the DD parameters. The D-optimal pattern, unburdened by the coprime constraints, best tracks the variation of $\boldsymbol{\Sigma}_\theta$ and achieves the best alignment. Remarkably, the proposed 2D Coprime pattern demonstrates exceptional flexibility besides fulfilling the coprime condition, achieving comparable performance to the D-optimal pattern. 

\begin{figure}[tbp] 
\leftskip=-0.3pt
\includegraphics[width=3.5in]{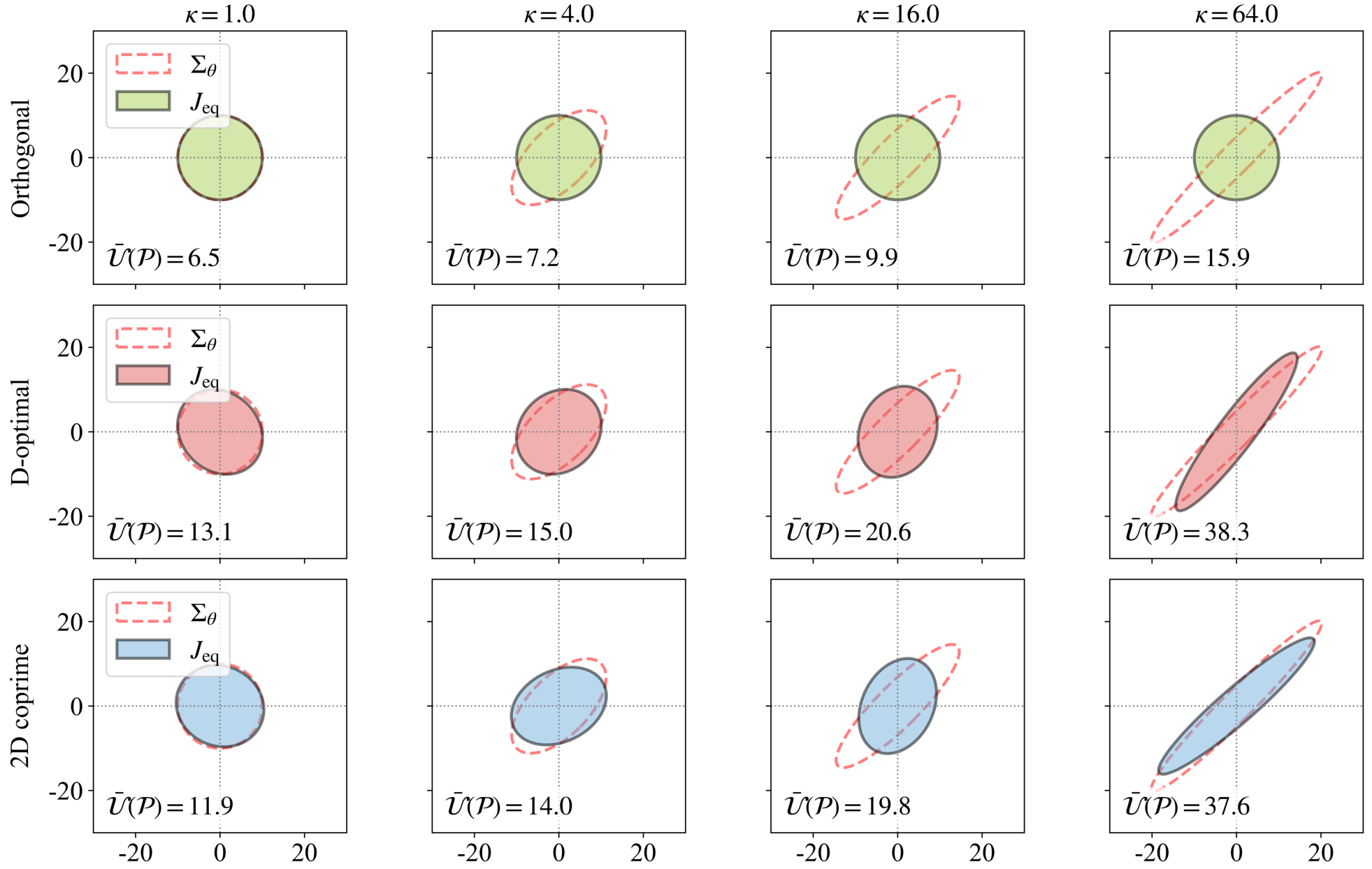}
\caption{Visualization of the equivalent FIM $\mathbf{J}_{\text{eq}}$ of different pilot patterns aligning with $\boldsymbol{\Sigma}_\theta$ with different condition numbers $\kappa$.}
\label{fig:Ellipse}
\end{figure}

Next, we evaluate the robustness of sensing accuracy under various covariances $\boldsymbol{\Sigma}_\theta$ to demonstrate the structural flexibility of the 2D Coprime condition. In Fig.~\ref{fig:bounds}a, we first demonstrate the robustness of different patterns by sweeping the condition number $\kappa$. The proposed 2D Coprime pattern consistently approaches the optimal accuracy of the D-optimal pattern across the entire span of $\kappa$. In contrast, the Orthogonal Coprime pattern exhibits a substantial MSE loss. This is due to its rigid orthogonal structure that limits the freedom for optimizing sensing accuracy.

Fig.~\ref{fig:bounds}b illustrates the robustness of patterns by sweeping the rotation angle $\phi$. As the angle rotates, the MSE curves exhibit periodic fluctuations. Notably, at certain mismatch angles (e.g., around $90^\circ\sim 150^\circ$), all patterns suffer from severe performance degradation. This occurs because the shape of the resource block inherently creates sensing blind spots at specific sensing directions in DD domain. The 2D Coprime pattern maintains a robust trajectory that approaches the D-optimal limit across various rotation angles.
In summary, the 2D Coprime pattern offers more freedom than the Orthogonal Coprime pattern to optimize sensing accuracy under various DD distributions. In other words, the 2D Coprime condition grants the system the essential freedom to achieve near-optimal sensing accuracy while maintaining the maximum unambiguous sensing region.

\section{Conclusion}
This paper proposed a coprime pilot design framework for DD parameter sensing in OFDM-ISAC systems. Based on the generalized Bézout identity, we first derived the 2D Coprime condition for achieving the maximum unambiguous sensing region. As far as the authors know, this is the first necessary and sufficient condition for unambiguous DD sensing in the literature. With the proposed condition, sensing accuracy was further optimized by leveraging the BCRB surrogate via the proposed AGIR algorithm. Numerical results demonstrated that the proposed pilot patterns outperform existing designs in terms of unambiguous sensing region and sensing accuracy. The superiority stems from the structural flexibility of the proposed condition to optimize sensing accuracy for different DD parameter distributions, where the maximum unambiguous sensing region is strictly guaranteed. Future work includes the extension of the proposed pilot design approach to other ISAC applications, such as semantic sensing.

\begin{figure}[tbp] 
\leftskip=-0.pt
\includegraphics[width=3.5in]{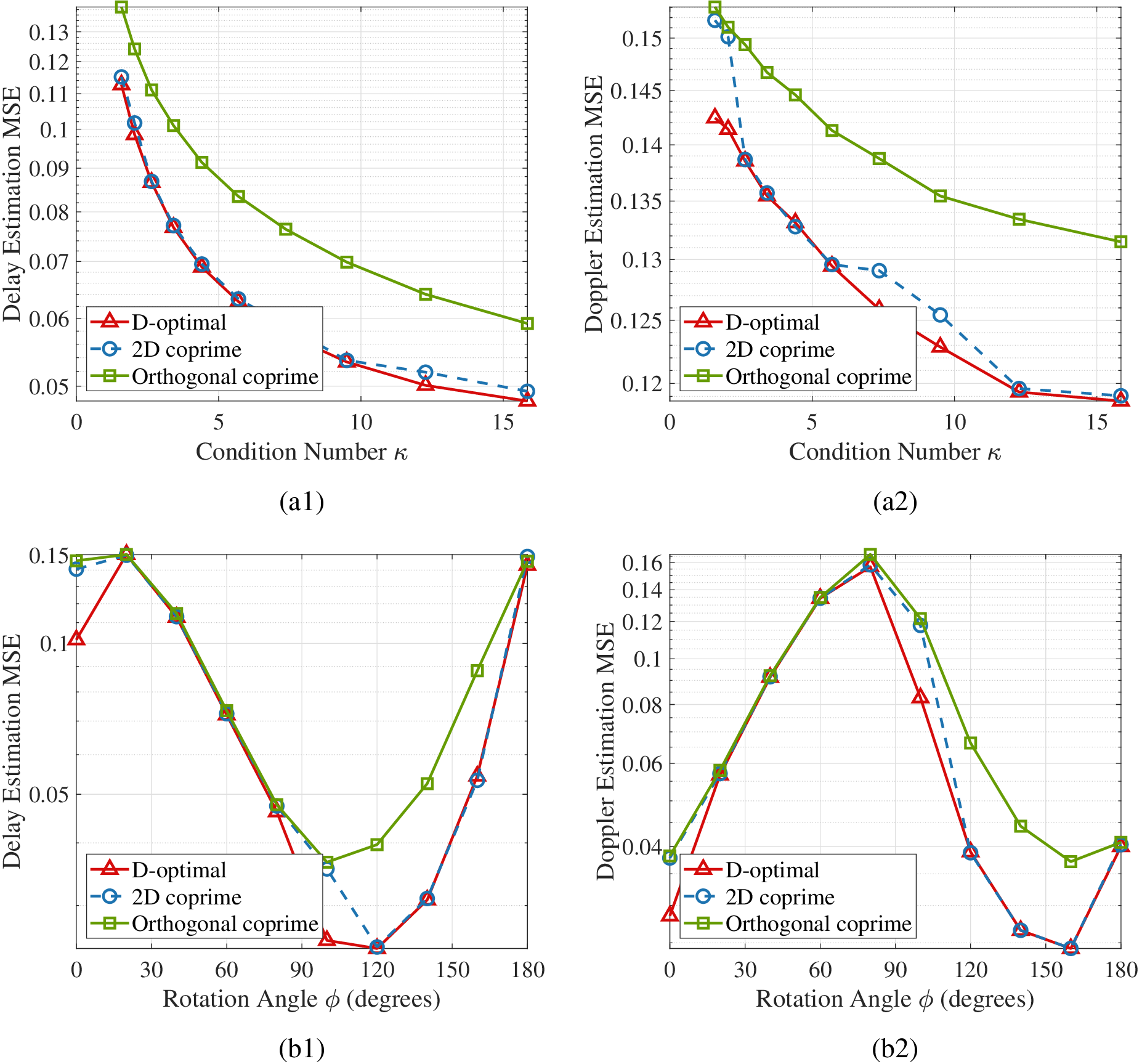}
\caption{ The robustness of sensing accuracy with different pilot patterns under varying (a) condition number $\kappa$ and (b) rotation angle $\phi$. }
\label{fig:bounds}
\end{figure}

\appendices

\section{Proof of Theorem~1}
\label{app:proof_thm1}

The unique recovery of $(q_\tau, q_\nu)$ is equivalent to determining whether the system of linear equations admits integer solutions that isolate the target variables.
For example, to solve $q_\tau$ in \eqref{sunau1}, we need to reach the unit vector $[1, 0]^T$ in the integer span:  
\begin{equation}\label{Ap1}
\sum_{i=1}^{R_d} \alpha_i \begin{bmatrix} u_i \\ v_i \end{bmatrix} 
+ k_1 \begin{bmatrix} N_c \\ 0 \end{bmatrix} 
+ k_2 \begin{bmatrix} 0 \\ N_s \end{bmatrix} 
= \begin{bmatrix} 1 \\ 0 \end{bmatrix}, \quad\alpha_i, k_1, k_2 \in \mathbb{Z}. 
\end{equation}

According to the generalized Bézout's identity, the solvability of the linear Diophantine system in \eqref{Ap1} requires that the integer lattice generated by its constituent vectors covers the entire $\mathbb{Z}^2$ space \cite{newman1972integral,cassels1997introduction}. To algebraically characterize this lattice, we define an augmented matrix $\mathbf{M} \in \mathbb{Z}^{(R_d+2) \times 2}$ whose rows consist of constituent vectors as: 
\begin{equation}
\mathbf{M} = 
\begin{bmatrix}
u_1     & v_1     \\
\vdots  & \vdots  \\
u_{R_d} & v_{R_d} \\
N_c     & 0       \\
0       & N_s
\end{bmatrix} .
\end{equation} 
Based on the Smith normal form of integer matrices\cite{smith1861xv}, the requirement that the row space of $\mathbf{M}$ encompasses $\mathbb{Z}^2$ (and thus explicitly contains the standard basis vectors $ [1, 0]^T$ and $ [0, 1]^T$) is satisfied if and only if the greatest common divisor of all $2 \times 2$ minors of $\mathbf{M}$ is equal to 1. The complete set of $2 \times 2$ minors of $\mathbf{M}$ can be given by 

\begin{enumerate}
\item \textit{Pilot-pilot minors} (rows $i$ and $j$, $i < j$):
   \begin{equation*}
   \det\begin{bmatrix} u_i & v_i \\ u_j & v_j \end{bmatrix} = u_i v_j - u_j v_i = \det(\mathbf{P}_{ij}). 
   \end{equation*}

\item \textit{Pilot-periodicity minors}:
   \begin{align*}
   \det\begin{bmatrix} u_i & v_i \\ N_c & 0 \end{bmatrix} &= -v_i N_c, \\
   \det\begin{bmatrix} u_i & v_i \\ 0 & N_s \end{bmatrix} &= u_i N_s.
   \end{align*}

\item \textit{Periodicity-periodicity minor}:
   \begin{equation*}
   \det\begin{bmatrix} N_c & 0 \\ 0 & N_s \end{bmatrix} = N_c N_s.
   \end{equation*}
\end{enumerate}
Therefore, the necessary and sufficient condition for the lattice to span $\mathbb{Z}^2$ (and hence unique resolution of $(q_\tau, q_\nu)$) is given by:  
\begin{equation}
\gcd\!\left( \bigcup_{i,j} \{ u_i v_j - u_j v_i \} \cup \bigcup_i \{ u_i N_s, v_i N_c \} \cup \{ N_c N_s \} \right) = 1,  
\end{equation}
which completes the proof of Theorem~\ref{theorem2}. 

In the context of pilot design, the linear combinations of relative pilot intervals must be able to synthesize the most refined resolution grid of symbols and subcarriers. Under such circumstances, the pilot pattern will completely untangle the coupled delay-Doppler phase and resolve sensing ambiguity.

\section{SMI Computation}
\label{app:smi_computation}

This appendix details the numerical computation of the SMI. 
For mathematical convenience, we first vectorize the transmitted and received pilot signal. Denote $\mathbf{y_\text{RS}} = \text{vec}(\mathbf{Y}_\text{RS}) \in \mathbb{C}^{D}$, where $D = N_c N_s$, and similarly $\mathbf{x}_\text{RS} = \text{vec}(\mathbf{X}_\text{RS})$. 
The SMI between the received signal $\mathbf{y_\text{RS}}$ and the channel parameters $\boldsymbol{\theta}$ is defined as:  
\begin{equation}
I(\mathbf{y_\text{RS}}; \boldsymbol{\theta}) = h(\mathbf{y_\text{RS}}) - h(\mathbf{y_\text{RS}} | \boldsymbol{\theta}). 
\end{equation}
Since direct calculation of the SMI is intractable, we evaluate it by separately computing its two constituent components: the conditional entropy $h(\mathbf{y_\text{RS}} | \boldsymbol{\theta})$ and marginal entropy $h(\mathbf{y_\text{RS}})$. To this end, we first establish the conditional distribution of the received signal $\mathbf{y_\text{RS}}$, which are necessary for calculating the concerned entropies. Subsequently, we explicitly demonstrate how to compute the conditional entropy and evaluate the marginal entropy via Monte Carlo estimation.

\subsubsection{Conditonal Distribution}
Let $\mathbf{A}(\tau, \nu) \in \mathbb{C}^{D \times 1}$ represent the steering vector for the given DD parameters $\boldsymbol{\theta} = [\tau, \nu]^T$, whose $d$-th element is defined as:   
\begin{equation}
[\mathbf{A}(\tau, \nu)]_{d} =  e^{-j2\pi (m \Delta f \tau - n T_s \nu)}, \ d = m + N_c(n-1). 
\end{equation}
The received signal $\mathbf{y_\text{RS}}$ follows a complex Gaussian distribution for the given DD parameters $\boldsymbol{\theta}$:  
\begin{equation}\label{guassdisappd}
\mathbf{y_\text{RS}} | \boldsymbol{\theta} \sim \mathcal{CN}\left(\mathbf{0}, \mathbf{\Sigma_{y_\text{RS}}}(\boldsymbol{\theta})\right),
\end{equation}
where the conditional covariance matrix is given by:  
\begin{equation}
\mathbf{\Sigma_{y_\text{RS}}}(\boldsymbol{\theta}) = \sigma^2_{\xi}  (\mathbf{A}(\tau, \nu) \odot \mathbf{x_\text{RS}})(\mathbf{A}(\tau, \nu) \odot \mathbf{x_\text{RS}})^H + \sigma^2 \mathbf{I}_D.
\end{equation}
Consider the DD parameters $\boldsymbol{\theta}$ over a finite lattice grid:  
\begin{itemize}
\item $\tau \in \{\tau^{(1)}, \tau^{(2)}, \dots, \tau^{(L_\tau)}\}$ with $L_\tau$ discrete values.
\item $\nu \in \{\nu^{(1)}, \nu^{(2)}, \dots, \nu^{(L_\nu)}\}$ with $L_\nu$ discrete values. 
\end{itemize}
As a result, the total number of possible channel states is given by $T = L_\tau \times L_\nu$, and each state occurs with equal probability $1/T$.

\subsubsection{Conditional Entropy Calculation}

\begin{algorithm}[t]
\caption{Monte Carlo Estimation of SMI}
\label{alg:smi_estimation}
\begin{algorithmic}[0]
\Require $N_c, N_s, \xi, \sigma^2, \mathbf{x}_{\text{RS}}, \{\tau^{(i)}\}_{i=1}^{L_\tau}, \{\nu^{(j)}\}_{j=1}^{L_\nu}, M$
\Ensure Estimated mutual information $\hat{I}(\mathbf{y}_{\text{RS}}; \boldsymbol{\theta})$
\Statex \textbf{Step 1: Pre-computation for each channel state}
\For{$t = 1$ \textbf{to} $T = L_\tau \times L_\nu$}
    \State Construct steering vector $\mathbf{a}_t = \mathbf{A}(\boldsymbol{\tau}_t, \boldsymbol{\nu}_t)$ and $\mathbf{A}_t = \mathbf{a}_t \odot \mathbf{x}_{\text{RS}}$, where {$\boldsymbol{\tau}_t, \boldsymbol{\nu}_t$ are the $t$-th lattice on the grid}
    \State Compute $\mathbf{\Sigma}_t =\sigma^2_{\xi}  \mathbf{A}_t \mathbf{A}_t^H + \sigma^2 \mathbf{I}_D$ \quad ($D = N_c N_s$) and store $\mathbf{\Sigma}_t^{-1}$ 
    \State Compute Cholesky decomposition $\mathbf{\Sigma}_t = \mathbf{L}_t \mathbf{L}_t^H$, store $\mathbf{L}_t$ and $\log\det(\mathbf{\Sigma}_t)$
\EndFor
\Statex \textbf{Step 2: Generate Monte Carlo samples}
\For{$t = 1$ \textbf{to} $T$}
    \For{$m = 1$ \textbf{to} $M$}
        \State Sample $\mathbf{z} \sim \mathcal{CN}(\mathbf{0}, \mathbf{I}_D)$
        \State Compute $\mathbf{y}_{\text{RS}, t,m} = \mathbf{L}_t \mathbf{z}$
        \State Store $\mathbf{y}_{\text{RS}, t,m}$ in $\mathbf{y}_{\text{RS}, \text{samples}}$
    \EndFor
\EndFor
\Statex \textbf{Step 3: Estimate marginal entropy}
\State Initialize $\text{sum\_log\_density} = 0$
\For{$i = 1$ \textbf{to} $TM$}
    \State $\mathbf{y}_i = \mathbf{y}_{\text{RS}, \text{samples}}[i]$
    \For{$t = 1$ \textbf{to} $T$}
        \State $\ell_t = -D\log\pi - \log\det(\mathbf{\Sigma}_t) - \mathbf{y}_i^H \mathbf{\Sigma}_t^{-1} \mathbf{y}_i$
    \EndFor
    \State $\ell_{\max} = \max(\ell_1, \dots, \ell_T)$
    \State $\log \hat{p}(\mathbf{y}_i) = \ell_{\max} + \log\!\left( \sum_{t=1}^{T} e^{(\ell_t - \ell_{\max})} \right) - \log T$
    \State $\text{sum\_log\_density} \gets \text{sum\_log\_density} + \log \hat{p}(\mathbf{y}_i)$
\EndFor
\State $\hat{h}(\mathbf{y}_{\text{RS}}) = -\frac{1}{TM} \cdot \text{sum\_log\_density}$ (nats)
\Statex \textbf{Step 4: Compute conditional entropy }
\State $\hat{h}(\mathbf{y}_{\text{RS}} \mid \boldsymbol{\theta}) = \frac{1}{T} \sum_{t=1}^{T} \log\!\left( (\pi e)^D \det(\mathbf{\Sigma}_t) \right)$
\Statex \textbf{Step 5: Compute SMI}
\State $\hat{I}(\mathbf{y}_{\text{RS}}; \boldsymbol{\theta}) = \hat{h}(\mathbf{y}_{\text{RS}}) - \hat{h}(\mathbf{y}_{\text{RS}} \mid \boldsymbol{\theta})$
\State \Return $\hat{I}(\mathbf{y}_{\text{RS}}; \boldsymbol{\theta})$
\end{algorithmic}
\end{algorithm}

The conditional entropy $h(\mathbf{y_\text{RS}} | \boldsymbol{\theta})$ can be computed as:
\begin{equation}
\begin{aligned}
h(\mathbf{y_\text{RS}} | \boldsymbol{\theta}) &= \mathbb{E}_{\boldsymbol{\theta}} \left[ \log \left( (\pi e)^D \det(\mathbf{\Sigma_{\mathbf{y_\text{RS}}}}(\boldsymbol{\theta})) \right) \right] \\
&= \frac{1}{T} \sum_{t=1}^{T} \log \left( (\pi e)^D \det(\mathbf{\Sigma}_t) \right),
\end{aligned}  
\end{equation}  
where $\mathbf{\Sigma}_t = \mathbf{\Sigma_{\mathbf{y_\text{RS}}}}(\boldsymbol{\theta}_t)$ denotes the $t$-th channel state.

\subsubsection{Marginal Entropy approximation}
To obtain the marginal entropy $h(\mathbf{y_\text{RS}})$, the computation for the following expectation is required:  
\begin{equation}\label{expectation}
h(\mathbf{y_\text{RS}}) = -\mathbb{E}_{\mathbf{y_\text{RS}}} \left[ \log p(\mathbf{y_\text{RS}}) \right], 
\end{equation}
where the PDF $p(\mathbf{y_\text{RS}})$ is given by:  
\begin{equation}
\label{mashijuli}
p(\mathbf{y_\text{RS}}) = \frac{1}{T} \sum_{t=1}^{T} p(\mathbf{y_\text{RS}} | \boldsymbol{\theta}_t) = \frac{1}{T} \sum_{t=1}^{T} \frac{e^{\left(-\mathbf{y_\text{RS}}^H \mathbf{\Sigma}_t^{-1} \mathbf{y_\text{RS}}\right)}}{\pi^D \det(\mathbf{\Sigma}_t)} .
\end{equation}
Unfortunately, the expectation in \eqref{expectation} lacks a closed-form solution. We thus evaluate it via Monte Carlo simulations. Specifically, Cholesky decomposition is utilized to factorize the covariance matrix as $\mathbf{\Sigma}_t = \mathbf{L}_t \mathbf{L}_t^H$, enabling efficient generation of high-dimensional samples $\mathbf{y}_{\text{RS}, t} = \mathbf{L}_t \mathbf{z}$ from normal vectors $\mathbf{z}$. Furthermore, we employ the Log-Sum-Exp for the marginal PDF in \eqref{mashijuli} to perform the accumulation in the logarithmic domain and prevent floating-point underflow. 

With the marginal and conditional entropies evaluated, the SMI is readily computed as their difference. The numerical method for computing SMI is outlined in Algorithm~\ref{alg:smi_estimation}.

 \bibliographystyle{IEEEtran}
\bibliography{refs}

\end{document}